\documentclass[a4paper]{article}
\usepackage{graphicx}
\usepackage{amsthm}
\usepackage{geometry}
\usepackage{microtype}
\newtheorem{theorem}{Theorem}[section]
\newtheorem{definition}[theorem]{Definition}
\newtheorem{lemma}[theorem]{Lemma}
\newtheorem{corollary}[theorem]{Corollary}
\newtheorem{property}[theorem]{Property}
\newenvironment{xrefthm}[1]{%
    \def\thexref{\ref{#1}} \begin{thexrefthm}%
    }{%
    \end{thexrefthm}%
}
\newtheorem*{thexrefthm}{Theorem \thexref}
\usepackage{lscape}

\catcode`\@=11
\def \myverbatimindent {3em}
\let \saveverbatime \@xverbatim
\def \@xverbatim {\leftskip = \myverbatimindent\relax\saveverbatime}
\catcode`\@=12

\usepackage{mathpartir}

\usepackage{extpfeil}

\usepackage[english]{babel}

\usepackage{enumitem}
\setlist{nosep}

\usepackage{multirow}
\usepackage{booktabs} 

\usepackage[unicode=true, pdfborder={0 0 0}]{hyperref}

\makeatletter
\newcommand{\prw@zbreak}{\nobreak\hskip\z@skip}
\newcommand{\BreakableHyphen}{\leavevmode%
  \prw@zbreak-\discretionary{}{}{}\prw@zbreak}
\DeclareRobustCommand{\hyp}{%
  \ifmmode-\else\BreakableHyphen\fi}
\makeatother

\newcommand{\ssi}{\ensuremath{\quad\text{iff}\quad}}

\def\vcentcolon{\mathrel{\mathop\ordinarycolon}}
\newcommand{\coloneqq}{\vcentcolon\mkern-1.2mu=}
\newcommand{\Coloneqq}{\vcentcolon\mkern-.9mu\vcentcolon\mkern-1.2mu=}

\newcommand{\push}{\ensuremath{\mathop{\mathbf{push}}\nolimits}}
\newcommand{\invoke}{\ensuremath{\mathop{\mathbf{invoke}}\nolimits}}
\newcommand{\patch}{\invoke}

\newcommand{\ajout}[3]{\ensuremath{#1 + \{#2\mapsto#3\} }}
\newcommand{\subst}[3]{\ajout{#1}{#2}{#3}}

\newcommand{\lift}[1]{\ensuremath{\mathop{(#1)}\nolimits_{\ast}}}

\newcommand{\dom}{\mathop{\mathrm{dom}}\nolimits}
\newcommand{\gc}[2]{#2 \setminus #1}
\newcommand{\Image}{\mathop{\mathrm{Im}}\nolimits}
\newcommand{\Loc}{\mathop{\mathrm{Loc}}\nolimits}
\newcommand{\Env}{\mathop{\mathrm{Env}}\nolimits}
\newcommand{\close}[1]{#1_{*}}
\newcommand{\extends}{\sqsupseteq}
\newcommand{\revextends}{\sqsubseteq}
\newcommand{\concat}{\cdot}
\newcommand{\cconcat}{\cdot}

\newcommand{\range}[3]{\ensuremath{#1_{#2}\dotsc#1_{#3}}}
\newcommand{\drange}[4]{\ensuremath{(#1_{#3},#2_{#3})\concat\dotsc\concat(#1_{#4},#2_{#4})}}

\newcommand{\unit}{\ensuremath{\mathbf{1}}}
\newcommand{\true}{\ensuremath{\mathop{\mathbf{true}}\nolimits}}
\newcommand{\false}{\ensuremath{\mathop{\mathbf{false}}\nolimits}}
\newcommand{\ite}[3]{\ensuremath{\mathop{\mathbf{if}}\ #1\
 \mathop{\mathbf{then}}\ #2\ \mathop{\mathbf{else}}\ #3}}
\newcommand{\letrec}[3]{\ensuremath{\mathop{\mathbf{letrec}}\ #1=#2\
  \mathop{\mathbf{in}}\ #3}}
\newcommand{\expr}{\ensuremath{\mathop{expr}\nolimits\/}}
\newcommand{\F}{\ensuremath{\mathcal{F}}}
\newcommand{\fun}[3]{\ensuremath{\left[\lambda#1.#2,#3\right]}}
\newcommand{\rhot}{\ensuremath{{\rho_{T}}}}
\newcommand{\env}[2]{\ensuremath{#1|#2}}

\newextarrow{\intArrow}{55{40}0}{\equiv\equiv\Rrightarrow}
\newextarrow{\optArrow}{55{40}0}{\Relbar\Relbar\Rightarrow}
\newextarrow{\naiveArrow}{55{40}0}{\relbar\relbar\rightarrow}

\newcommand{\reduction}[5][\env{\rhot}{\rho}]{\ensuremath{#2^{\,#3}
  \optArrow[\F]{#1} #4^{\,#5}}}
\newcommand{\reductionF}[6][\env{\rhot}{\rho}]{\ensuremath{#2^{\,#3}
  \optArrow[{#6}]{#1} #4^{\,#5}}}
\newcommand{\reductionR}[5][\env{\rhot}{\rho}]{\ensuremath{#2^{\,#3}
  \optArrow[\close{\F}]{#1} #4^{\,#5}}}
\newcommand{\reduclift}[5][\env{\rhot}{\rho}]{\ensuremath{{\lift{#2}}^{\,#3}
  \optArrow[\lift{\F}]{#1} #4^{\,#5}}}

\newcommand{\reductionI}[5][\env{\rhot}{\rho}]{\ensuremath{#2^{\,#3}
  \intArrow[\F]{#1} #4^{\,#5}}}
\newcommand{\reductionIF}[6][\env{\rhot}{\rho}]{\ensuremath{#2^{\,#3}
  \intArrow[{#6}]{#1} #4^{\,#5}}}

\newcommand{\reductionN}[5][\rho]{\ensuremath{#2^{\,#3}
  \naiveArrow[\F]{#1} #4^{\,#5}}}
\newcommand{\reductionNF}[6][\rho]{\ensuremath{#2^{\,#3}
  \naiveArrow[#6]{#1} #4^{\,#5}}}

\begin{document}

\title{Lambda-lifting and CPS conversion in an imperative language}

\author{Gabriel Kerneis       \and
        Juliusz Chroboczek \and
        {\em Universit\'e Paris Diderot, PPS, Paris, France} 
}

\date{February 2012}

\maketitle

\begin{abstract}
This paper is a companion technical report to the article ``Continuation-Passing
C: from threads to events through continuations''.  It contains the complete
version of the proofs of correctness of lambda-lifting and CPS-conversion
presented in the article.
\end{abstract}

\tableofcontents

\section{Introduction}

This paper is a companion technical report to the article ``Continuation-Passing
C: from threads to events through continuations'' \cite{cpc2012}.  It contains
the complete version of the proofs presented in the article.  It does not,
however, give any background or motivation for our work: please refer to the
original article.

\section{Lambda-lifting in an imperative language}
\label{sec:lifting}

To prove the correctness of lambda-lifting in an imperative,
call-by-value language when functions are called in tail position, we do
not reason directly on CPC programs, because the semantics of C is too
broad and complex for our purposes.  The CPC translator leaves most
parts of converted programs intact, transforming only control structures
and function calls.  Therefore, we define a simple language with
restricted values, expressions and terms, that captures the features we
are most interested in (Section~\ref{sec:definitions}).

The reduction rules for this language (Section~\ref{sec:naive-def}) use
a simplified memory model without pointers and enforce that local
variables are not accessed outside of their scope, as ensured by our
boxing pass.  This is necessary since lambda-lifting
is not correct in general in the presence of extruded variables.

It turns out that the ``naive'' reduction rules defined in
Section~\ref{sec:naive-def} do not provide strong enough invariants to
prove this correctness theorem by induction, mostly because we represent
memory with a store that is not invariant with respect to lambda-lifting.
Therefore, in Section~\ref{sec:semopt}, we define an equivalent,
``optimised'' set of reduction rules which enforces more regular stores
and closures.

The proof of correctness is then carried out in
Section~\ref{sec:correction-ll} using these optimised rules.  We first
define the invariants needed for the proof and formulate a strengthened
version of the correctness theorem (Theorem~\ref{thm:correction-ll},
Section~\ref{sec:strong-invariants}).  A comprehensive overview of the
proof is then given in Section~\ref{sec:overview}.  The proof is fully
detailed in Section~\ref{sec:proof-correctness}, with the help of a number
of lemmas to keep the main proof shorter
(Sections~\ref{sec:rewriting-lemmas} and~\ref{sec:aliasing-lemmas}).

The main limitation of this proof is that
Theorems~\ref{thm:lambda-lifting-correctness}
and~\ref{thm:correction-ll} are implications, not equivalences: we do
not prove that if a term does not reduce, it will not reduce once
lifted.  For instance, this proof does not ensure that lambda-lifting
does not break infinite loops.

\subsection{Definitions}
\label{sec:definitions}

In this section, we define the terms
(Definition~\ref{def:full-language}), the reduction rules
(Section~\ref{sec:naive-def}) and the lambda-lifting transformation
itself (Section~\ref{sec:lifting-def}) for our small imperative
language.  With these preliminary definitions, we are then able to
characterise \emph{liftable parameters}
(Definition~\ref{dfn:var-liftable-simple}) and state the main
correctness theorem (Theorem~\ref{thm:lambda-lifting-correctness},
Section~\ref{sec:correctness}).

\begin{definition}[Values, expression and terms]\label{def:full-language}

Values are either boolean and integer constants or $\unit$, a special
value for functions returning \texttt{void}.
\[v \Coloneqq \quad\unit \;|\; \true \;|\; \false \;|\; n \in \mathbf{N}\]

Expressions are either values or variables.  We deliberately omit
arithmetic and boolean operators, with the sole concern of avoiding
boring cases in the proofs.
\[e \Coloneqq \quad v \;|\; x\;|\; \dotsc\]

Terms are consist of assignments, conditionals, sequences, recursive
functions definitions and calls.
    \begin{align*}
    T \Coloneqq & \quad e
        \;|\; x \coloneqq T
        \;|\; \ite{T}{T}{T}
        \;|\; T\ ;\ T \\
        \;|\; & \letrec{f(\range{x}{1}{n})}{T}{T}
        \;|\; f(T,\dotsc,T)
        \end{align*}
        \qedhere
\end{definition}
Our language focuses on the essential details affected by the
transformations: recursive functions, conditionals and memory accesses.
Loops, for instance, are ignored because they can be expressed in terms
of recursive calls and conditional jumps --- and that is, in fact, how
the splitting pass translates them.
Since lambda-lifting happens after the splitting pass, our language
need to include inner functions (although they are not part of the C
language), but it can safely exclude \texttt{goto} statements.

\subsubsection{Naive reduction rules\label{sec:naive-def}}

\paragraph{Environments and stores}
Handling inner functions requires explicit closures in the reduction
rules.  We need environments, written $\rho$, to bind variables to
locations, and a store, written $s$, to bind locations to values.

\emph{Environments} and \emph{stores} are partial functions, equipped
with a single operator which extends and modifies a partial function:
\ajout{\cdot}{\cdot}{\cdot}.

\begin{definition} The modification (or extension) $f'$ of a partial
function $f$, written $f' = \ajout{f}{x}{y}$, is defined as follows:
\begin{align*}
f'(t) =& \begin{cases}
y&\text{when $t$ = $x$}\\
f(t)&\text{otherwise}
\end{cases}\\
\dom(f') =& \dom(f)\cup\{x\}
\end{align*}
\qedhere
\end{definition}

\begin{definition}[Environments of variables and functions]
Environments of variables are defined inductively by
\[\rho \Coloneqq \varepsilon \;|\; (x,l)\concat\rho,\]
i.e.\ the empty domain function and $\ajout{\rho}{x}{l}$ (respectively).

Environments of functions associate function names to closures:
\[\F : \{f, g, h, \dotsc\} \rightarrow
    \{\fun{\range{x}{1}{n}}{T}{\rho,\F}\}.\]
\qedhere
\end{definition}

Note that although we have a notion of locations, which correspond
roughly to memory addresses in C, there is no way to copy, change or
otherwise manipulate a location directly in the syntax of our language.
This is on purpose, since adding this possibility would make
lambda-lifting incorrect: it translates the fact, ensured by the boxing
pass in the CPC translator, that there are no extruded variables in the
lifted terms.

\paragraph{Reduction rules}
We use classical big-step reduction rules for our language
(Figure~\ref{sem-proof:naive}, p.~\pageref{sem-proof:naive}).

\begin{figure}
\begin{gather*}
  \inferrule*[Left=(val)]{ }{\reductionN{v}{s}{v}{s}} \qquad\qquad
  \inferrule*[Left=(var)]{\rho\ x = l \in \dom\ s}{\reductionN{x}{s}{s\ l}{s}}\\
  \inferrule*[Left=(assign)]{\reductionN{a}{s}{v}{s'} \\
    \rho\ x = l \in \dom\ s'}{\
	\reductionN{x \coloneqq a}{s}{\unit}{\subst{s'}{l}{v}}}\qquad\qquad
  \inferrule*[Left=(seq)]{\reductionN{a}{s}{v}{s'} \\ \reductionN{b}{s'}{v'}{s''}}{\reductionN{a\ ;\ b}{s}{v'}{s''}}\\
  \inferrule*[Left=(if-t.)]{\reductionN{a}{s}{\true}{s'} \\ \reductionN{b}{s'}{v}{s''}}{\
	\reductionN{\ite{a}{b}{c}}{s}{v}{s''}}\qquad\qquad
  \inferrule*[Left=(if-f.)]{\reductionN{a}{s}{\false}{s'} \\ \reductionN{c}{s'}{v}{s''}}{\
	\reductionN{\ite{a}{b}{c}}{s}{v}{s''}}\\
  \inferrule*[Left=(letrec)]{\
    \reductionNF{b}{s}{v}{s'}{\mathcal{F'}} \\\\
    \mathcal{F'}=\ajout{\F}{f}{\fun{\range{x}{1}{n}}{a}{\rho,\F}}    }{\
	\reductionN{\letrec{f(\range{x}{1}{n})}{a}{b}}{s}{v}{s'}}\\
  \inferrule*[Left=(call)]{\
\F\,f = \fun{\range{x}{1}{n}}{b}{\rho',\mathcal{F'}}\\
\rho''= \drange{x}{l}{1}{n}\\
\text{$l_{i}$ fresh and distinct}\\\\
\forall i,\reductionN{a_i}{s_i}{v_i}{s_{i+1}} \\
\reductionNF[\rho''\concat\rho']{b}{\ajout{s_{n+1}}{l_i}{v_i}}{v}{s'}{\ajout{\mathcal{F'}}{f}{\F\,f}}
}{\
\reductionN{f(\range{a}{1}{n})}{s_{1}}{v}{s'}}
\end{gather*}
\caption{``Naive'' reduction rules\label{sem-proof:naive}}
\end{figure}

In the (call) rule, we need to introduce \emph{fresh} locations for the
parameters of the called function.  This means that we must choose
locations that are not already in use, in particular in the environments
$\rho'$ and $\F$.  To express this choice, we define two ancillary
functions, $\Env$ and $\Loc$, to extract the environments and locations
contained in the closures of a given environment of functions $\F$.
\begin{definition}[Set of environments, set of locations]
  \[\Env(\F) = \bigcup\left\{ \rho, \rho'
  \ |\ \fun{\range{x}{1}{n}}{M}{\rho,\mathcal{F'}} \in
  \Image(\F), \rho' \in \Env(\mathcal{F'})\right\}\]
  \[\Loc(\F) = \bigcup\left\{ \Image(\rho)
  \ |\ \rho \in
  \Env(\F)\right\}\]
  \[\text{A location $l$ is said to \emph{appear} in } \F \ssi l \in
      \Loc(\F).\]
  \qedhere
\end{definition}
These functions allow us to define fresh locations.
\begin{definition}[Fresh location] In the (call) rule, a location is
\emph{fresh} when:
\begin{itemize}
  \item $l \notin \dom(s_{n+1})$, i.e.\ $l$ is not already used in the store
  before the body of $f$ is evaluated, and
  \item $l$ doesn't appear in $\ajout{\mathcal{F'}}{f}{\F\,f}$, i.e.\
  $l$ will not interfere with locations captured in the environment of
  functions.
\end{itemize}
\qedhere
\end{definition}
Note that the second condition implies in particular that $l$ does not
appear in either $\F$ or $\rho'$.

\subsubsection{Lambda-lifting}\label{sec:lifting-def}

Lambda-lifting can be split into two parts: parameter lifting and block
floating\cite{danvy}.  We will
focus only on the first part here, since the second one is trivial.
Parameter lifting consists in adding a free variable as a parameter of
every inner function where it appears free.  This step is repeated until
every variable is bound in every function, and closed functions can
safely be floated to top-level.  Note that although the transformation
is called lambda-lifting, we do not focus on a single function and try
to lift all of its free variables; on the contrary, we define the
lifting of a single free parameter $x$ in every possible function.

Smart lambda-lifting algorithms strive to minimize the number
of lifted variables.  Such is not our concern in this proof: parameters
are lifted in every function where they might potentially be free.

\begin{definition}[Parameter lifting in a term]\label{dfn:lifted-term}
Assume that $x$ is defined as a parameter of a given function $g$, and
that every inner function in $g$ is called $h_i$ (for some
$i\in\mathbf{N}$).  Also assume that function parameters are unique before
lambda-lifting.

\noindent Then the \emph{lifted form} $\lift{M}$ of the term $M$ with
respect to $x$ is defined inductively as follows:
{\allowdisplaybreaks
  \begin{gather*}
  \lift{\unit} = \unit \qquad \lift{n} = n \\
  \lift{true} = true \qquad \lift{false} = false \\
  \lift{y} = y \quad \text{ and } \quad \lift{y \coloneqq a}= y \coloneqq
  \lift{a} \quad \text{(even if $y=x$)} \\
  \lift{a\ ;\ b} = \lift{a}\ ;\ \lift{b} \\
  \lift{\ite{a}{b}{c}} = \ite{\lift{a}}{\lift{b}}{\lift{c}} \\
  \lift{ \letrec{f(\range{x}{1}{n})}{a}{b} } =
  \begin{cases}
  \letrec{f(\range{x}{1}{n}x)}{\lift{a}}{\lift{b}}  &\text{if $f = h_i$}\\
  \letrec{f(\range{x}{1}{n})}{\lift{a}}{\lift{b}}   &\text{otherwise}
  \end{cases}\\
  \lift{f(\range{a}{1}{n})} =
  \begin{cases}
  f(\lift{a_{1}},\dotsc,\lift{a_{n}},x)&\text{if $f = h_i$ for some $i$}\\
  f(\lift{a_{1}},\dotsc,\lift{a_{n}})&\text{otherwise}
  \end{cases}
  \end{gather*}
}
  \qedhere
\end{definition}

\subsubsection{Correctness condition}\label{sec:correctness}

We show that parameter lifting is correct for variables defined in
functions whose inner functions are called exclusively in \emph{tail
position}.  We call these variables \emph{liftable parameters}.

We first define tail positions as usual \cite{clinger}:
\begin{definition}[Tail position]
\emph{Tail positions} are defined inductively as follows:
\begin{enumerate}
\item $M$ and $N$ are in tail position in \ite{P}{M}{N}.
\item $N$ is in tail position in $N$ and $M \ ;\ N$ and \letrec{f(\range{x}{1}{n})}{M}{N}.
\end{enumerate}
\qedhere
\end{definition}
A parameter $x$ defined in a function $g$ is liftable if every inner
function in $g$ is called exclusively in tail position.
\begin{definition}[Liftable parameter] \label{dfn:var-liftable-simple}
A parameter $x$ is \emph{liftable} in $M$ when:
\begin{itemize}
\item $x$ is defined as the parameter of a function $g$,
\item inner functions in $g$, named $h_i$, are called exclusively in
tail position in $g$ or in one of the $h_i$.
\end{itemize}
\qedhere
\end{definition}
Our main theorem states that performing parameter-lifting on a liftable
parameter preserves the reduction:
\begin{theorem}[Correctness of lambda-lifting]
\label{thm:lambda-lifting-correctness}
If $x$ is a liftable parameter in $M$, then
\[\exists t,
\reductionNF[\varepsilon]{M}{\varepsilon}{v}{t}{\varepsilon} \text{ implies }
\exists t',
\reductionNF[\varepsilon]{\lift{M}}{\varepsilon}{v}{t'}{\varepsilon}.\]
\end{theorem}
Note that the resulting store $t'$ changes because lambda-lifting
introduces new variables, hence new locations in the store, and changes
the values associated with lifted variables;
Section~\ref{sec:correction-ll} is devoted to the proof of this theorem.
To maintain invariants during the proof, we need to use an equivalent,
``optimised'' set of reduction rules; it is introduced in the next section.

\subsection{Optimised reduction rules\label{sec:semopt}}

The naive reduction rules (Section~\ref{sec:naive-def}) are not
well-suited to prove the correctness of lambda-lifting.  Indeed, the
proof is by induction and requires a number of invariants on the
structure of stores and environments.  Rather than having a dozen of
lemmas to ensure these invariants during the proof of correctness, we
translate them as constraints in the reduction rules.

To this end, we introduce two optimisations --- minimal stores
(Section~\ref{sec:mini-store}) and compact closures
(Section~\ref{sec:compact-closures}) --- which lead to the definition of
an optimised set of reduction rules (Figure~\ref{sem-proof:opt},
Section~\ref{sec:opt-rules}).  The equivalence between optimised and
naive reduction rules is shown in Section~\ref{sec:sem-equiv}.

\subsubsection{Minimal stores} \label{sec:mini-store}

In the naive reduction rules, the store grows faster when reducing lifted
terms, because each function call adds to the store as many locations as it
has function parameters.  This yields stores of different sizes when
reducing the original and the lifted term, and that difference cannot be
accounted for locally, at the rule level.

Consider for instance the simplest possible case of lambda-lifting:
\begin{gather}
    \letrec{g(x)}{(\letrec{h()}{x}{h()})}{g(\unit)}\tag{original} \\
    \letrec{g(x)}{(\letrec{h(y)}{y}{h(x)})}{g(\unit)}\tag{lifted}
\end{gather}
At the end of the reduction, the store for the original term is
$\{l_x \mapsto \unit \}$ whereas the store for the lifted term is
$\{l_x \mapsto \unit ; l_y \mapsto \unit \}$.  More complex terms
would yield even larger stores, with many out-of-date copies of lifted
variables.

To keep the store under control, we need to get rid of useless variables
as soon as possible during the reduction.  It is safe to remove a
variable $x$ from the store once we are certain that it will never be
used again, i.e.\ as soon as the term in tail position in the function
which defines $x$ has been evaluated.  This mechanism is analogous
to the deallocation of a stack frame when a function returns.

To track the variables whose location can be safely reclaimed after the
reduction of some term $M$, we introduce \emph{split environments}.
Split environments are written $\env{\rhot}{\rho}$, where $\rhot$ is
called the \emph{tail environment} and $\rho$ the non-tail one; only the
variables belonging to the tail environment may be safely reclaimed.
The reduction rules build environments so that a variable $x$ belongs to
$\rhot$ if and only if the term $M$ is in tail position in the current
function $f$ and $x$ is a parameter of $f$.  In that case, it is safe to
discard the locations associated to all of the parameters of $f$,
including $x$, after $M$ has been reduced because we are sure that the
evaluation of $f$ is completed (and there are no first-class functions in
the language to keep references on variables beyond their scope of
definition).

We also define a  \emph{cleaning} operator, $\gc{\cdot}{\cdot}$, to
remove a set of variables from the store.
\begin{definition}[Cleaning of a store]
The store $s$ cleaned with respect to the variables in $\rho$, written
$\gc{\rho}{s}$, is defined as
$\gc{\rho}{s} = s |_{\dom(s)\setminus\Image(\rho)}$.
\qedhere
\end{definition}

\subsubsection{Compact closures} \label{sec:compact-closures}
Another source of complexity with the naive reduction rules is the inclusion of
useless variables in closures.  It is safe to remove from the
environments of variables contained in closures the variables that are also
parameters of the function: when the function is called, and the
environment restored, these variables will be hidden by the freshly
instantiated parameters.

This is typically what happens to lifted parameters: they are free
variables, captured in the closure when the function is defined, but
these captured values will never be used since calling the function adds
fresh parameters with the same names.  We introduce \emph{compact
closures} in the optimised reduction rules to avoid dealing with this
hiding mechanism in the proof of lambda-lifting.

A compact closure is a closure that does not capture any variable which
would be hidden when the closure is called because of function
parameters having the same name.
\begin{definition}[Compact closure and environment]
  A closure $\fun{\range{x}{1}{n}}{M}{\rho,\F}$ is  \emph{compact}
  if $\forall i, x_i\notin\dom(\rho)$ and \F\/ is compact.
  An environment is \emph{compact} if it contains only compact closures.
  \qedhere
\end{definition}
We define a canonical mapping from any environment $\F$ to a compact
environment $\close{\F}$, restricting the domains of every closure in
$\F$.
\begin{definition}[Canonical compact environment]
  The \emph{canonical compact environment} $\close{\F}$ is the
  unique environment with the same domain as $\F$ such that 
  \begin{align*}
    \forall f \in \dom(\F),
    \F\,f &= \fun{\range{x}{1}{n}}{M}{\rho,\mathcal{F'}}\\
    \text{implies }\close{\F}\ f &=
    \fun{\range{x}{1}{n}}{M}{\rho|_{\dom(\rho)\setminus\{\range{x}{1}{n}\}},\close{\mathcal{F'}}}.
  \end{align*}
    \qedhere
\end{definition}

\subsubsection{Optimised reduction rules} \label{sec:opt-rules}
Combining both optimisations yields the \emph{optimised} reduction rules
(Figure~\ref{sem-proof:opt}, p.~\pageref{sem-proof:opt}), used
Section~\ref{sec:correction-ll} for the proof of lambda-lifting.
We ensure minimal stores by cleaning them in the (val), (var) and
(assign) rules, which correspond to tail positions; split environments
are introduced in the (call) rule to distinguish fresh parameters, to be
cleaned, from captured variables, which are preserved.  Tail positions are
tracked in every rule through split environments, to avoid cleaning
variables too early, in a non-tail branch.

We also build compact closures in the (letrec) rule by removing the
parameters of $f$ from the captured environment $\rho'$.  

\begin{figure}[hbt]
\begin{gather*}
  \inferrule*[Left=(val)]{ }{\reduction{v}{s}{v}{\gc{\rhot}{s}}}\qquad\qquad
  \inferrule*[Left=(var)]{\rhot\concat\rho\ x = l \in \dom\ s}{\reduction{x}{s}{s\ l}{\gc{\rhot}{s}}}\\
  \inferrule*[Left=(assign)]{\reduction[\env{}{\rhot\concat\rho}]{a}{s}{v}{s'} \\
    \rhot\concat\rho\ x = l \in \dom\ s'}{\
	\reduction{x \coloneqq a}{s}{\unit}{\gc{\rhot}{\subst{s'}{l}{v}}}}\qquad\qquad
  \inferrule*[Left=(seq)]{\reduction[\env{}{\rhot\concat\rho}]{a}{s}{v}{s'} \\ \reduction{b}{s'}{v'}{s''}}{\reduction{a\ ;\ b}{s}{v'}{s''}}\\
  \inferrule*[Left=(if-t.)]{\reduction[\env{}{\rhot\concat\rho}]{a}{s}{\true}{s'} \\ \reduction{b}{s'}{v}{s''}}{\
	\reduction{\ite{a}{b}{c}}{s}{v}{s''}}\qquad\qquad
  \inferrule*[Left=(if-f.)]{\reduction[\env{}{\rhot\concat\rho}]{a}{s}{\false}{s'} \\ \reduction{c}{s'}{v}{s''}}{\
	\reduction{\ite{a}{b}{c}}{s}{v}{s''}}\\
  \inferrule*[Left=(letrec)]{\
    \reductionF{b}{s}{v}{s'}{\mathcal{F'}} \\\\
    \rho' = \rhot\concat\rho|_{\dom(\rhot\concat\rho)\setminus\{\range{x}{1}{n}\}} \\
    \mathcal{F'}=\ajout{\F}{f}{\fun{\range{x}{1}{n}}{a}{\rho',\F}}    }{\
	\reduction{\letrec{f(\range{x}{1}{n})}{a}{b}}{s}{v}{s'}}\\
  \inferrule*[Left=(call)]{\
\F\,f = \fun{\range{x}{1}{n}}{b}{\rho',\mathcal{F'}}\\
\rho''= \drange{x}{l}{1}{n}\\
\text{$l_{i}$ fresh and distinct}\\\\
\forall i,\reduction[\env{}{\rhot\concat\rho}]{a_i}{s_i}{v_i}{s_{i+1}} \\
\reductionF[\env{\rho''}{\rho'}]{b}{\ajout{s_{n+1}}{l_i}{v_i}}{v}{s'}{\ajout{\mathcal{F'}}{f}{\F\,f}}
}{\
\reduction{f(\range{a}{1}{n})}{s_{1}}{v}{\gc{\rhot}{s'}}}
\end{gather*}
\caption{Optimised reduction rules\label{sem-proof:opt}}
\end{figure}

\begin{theorem}[Equivalence between naive and optimised reduction rules]\label{thm:sem-equiv}
Optimised and naive reduction rules are equivalent: every reduction in
one set of rules yields the same result in the other.  It is necessary,
however, to take care of locations left in the store by the naive reduction:
  \[
  \reductionF[\env{\varepsilon}{\varepsilon}]{M}{\varepsilon}{v}{\varepsilon}{\varepsilon}
  \ssi
 \exists s,\reductionNF[\varepsilon]{M}{\varepsilon}{v}{s}{\varepsilon}
 \]
\end{theorem}
We prove this theorem in Section~\ref{sec:sem-equiv}.

\subsection{Equivalence of optimised and naive reduction rules\label{sec:sem-equiv}}

This section is devoted to the proof of equivalence between the optimised
naive reduction rules (Theorem~\ref{thm:sem-equiv}).

To clarify the proof, we introduce intermediate reduction rules
(Figure~\ref{sem-proof:inter}, p.~\pageref{sem-proof:inter}), with only
one of the two optimisations: minimal stores, but not compact closures.

The proof then consists in proving that optimised and intermediate rules
are equivalent (Lemma~\ref{lem:IimpliesO} and Lemma~\ref{lem:OimpliesI},
Section~\ref{subsec:first-step}), then that naive and intermediate rules
are equivalent  (Lemma~\ref{lem:IimpliesN} and Lemma~\ref{lem:NimpliesI},
Section~\ref{subsec:second-step}).

\[
\text{Naive rules}
\xtofrom[\text{Lemma~\ref{lem:IimpliesN}}]{\text{Lemma~\ref{lem:NimpliesI}}}
\text{Intermediate rules}
\xtofrom[\text{Lemma~\ref{lem:OimpliesI}}]{\text{Lemma~\ref{lem:IimpliesO}}}
\text{Optimised rules}
\]

\begin{figure}[hbt]
\begin{gather*}
  \inferrule*[Left=(val)]{ }{\reductionI{v}{s}{v}{\gc{\rhot}{s}}}\qquad\qquad
  \inferrule*[Left=(var)]{\rhot\concat\rho\ x = l \in \dom\ s}{\reductionI{x}{s}{s\ l}{\gc{\rhot}{s}}}\\
  \inferrule*[Left=(assign)]{\reductionI[\env{}{\rhot\concat\rho}]{a}{s}{v}{s'} \\
    \rhot\concat\rho\ x = l \in \dom\ s'}{\
	\reductionI{x \coloneqq a}{s}{\unit}{\gc{\rhot}{\subst{s'}{l}{v}}}}\qquad\qquad
  \inferrule*[Left=(seq)]{\reductionI[\env{}{\rhot\concat\rho}]{a}{s}{v}{s'} \\ \reductionI{b}{s'}{v'}{s''}}{\reductionI{a\ ;\ b}{s}{v'}{s''}}\\
  \inferrule*[Left=(if-t.)]{\reductionI[\env{}{\rhot\concat\rho}]{a}{s}{\true}{s'} \\ \reductionI{b}{s'}{v}{s''}}{\
	\reductionI{\ite{a}{b}{c}}{s}{v}{s''}}\qquad\qquad
  \inferrule*[Left=(if-f.)]{\reductionI[\env{}{\rhot\concat\rho}]{a}{s}{\false}{s'} \\ \reductionI{c}{s'}{v}{s''}}{\
	\reductionI{\ite{a}{b}{c}}{s}{v}{s''}}\\
  \inferrule*[Left=(letrec)]{\
    \reductionIF{b}{s}{v}{s'}{\mathcal{F'}} \\\\
    \rho' = \rhot\concat\rho \\
    \mathcal{F'}=\ajout{\F}{f}{\fun{\range{x}{1}{n}}{a}{\rho,\F}}    }{\
	\reductionI{\letrec{f(\range{x}{1}{n})}{a}{b}}{s}{v}{s'}}\\
  \inferrule*[Left=(call)]{\
\F\,f = \fun{\range{x}{1}{n}}{b}{\rho',\mathcal{F'}}\\
\rho''= \drange{x}{l}{1}{n}\\
\text{$l_{i}$ fresh and distinct}\\\\
\forall i,\reductionI[\env{}{\rhot\concat\rho}]{a_i}{s_i}{v_i}{s_{i+1}} \\
\reductionIF[\env{\rho''}{\rho'}]{b}{\ajout{s_{n+1}}{l_i}{v_i}}{v}{s'}{\ajout{\mathcal{F'}}{f}{\F\,f}}
}{\
\reductionI{f(\range{a}{1}{n})}{s_{1}}{v}{\gc{\rhot}{s'}}}
\end{gather*}
\caption{Intermediate reduction rules\label{sem-proof:inter}}
\end{figure}

\subsubsection{Optimised and intermediate reduction rules equivalence\label{subsec:first-step}}

In this section, we show that optimised and intermediate reduction rules
are equivalent:
\[
\text{Intermediate rules}
\xtofrom[\text{Lemma~\ref{lem:OimpliesI}}]{\text{Lemma~\ref{lem:IimpliesO}}}
\text{Optimised rules}
\]
We must therefore show that it is correct to use compact closures in the
optimised reduction rules.

Compact closures carry the implicit idea that some variables can be
safely discarded from the environments when we know for sure that they
will be hidden.  The following lemma formalises this intuition.
\begin{lemma}[Hidden variables elimination]\label{lem:intro-in-env2}
  \begin{align*}
  \forall l,l', \reductionI[\env{\rhot\concat(x,l)}{\rho}]{M}{s}{v}{s'} \ssi&
  \reductionI[\env{\rhot\concat(x,l)}{(x,l')\concat\rho}]{M}{s}{v}{s'}
\\
  \forall l,l', \reduction[\env{\rhot\concat(x,l)}{\rho}]{M}{s}{v}{s'} \ssi&
  \reduction[\env{\rhot\concat(x,l)}{(x,l')\concat\rho}]{M}{s}{v}{s'}
  \end{align*}
  Moreover, both derivations have the same height.
\end{lemma}
\begin{proof}
 The exact same proof holds for both intermediate and optimised reduction
 rules.

 By induction on the structure of the derivation.  The proof relies solely
 on the fact that $\rhot\concat(x,l)\concat\rho =
 \rhot\concat(x,l)\concat(x,l')\concat\rho$.

\paragraph{(seq)}
      $\rhot\concat(x,l)\concat\rho = \rhot\concat(x,l)\concat(x,l')\concat\rho$.
      So,
      \[\reductionI[\env{}{\rhot\concat(x,l)\concat(x,l')\concat\rho}]{a}{s}{v}{s'} \ssi
      \reductionI[\env{}{\rhot\concat(x,l)\concat\rho}]{a}{s}{v}{s'}\]
      Moreover, by the induction hypotheses,
      \[\reductionI[\env{\rhot\concat(x,l)}{(x,l')\concat\rho}]{b}{s'}{v'}{s''} \ssi
      \reductionI[\env{\rhot\concat(x,l)}{\rho}]{b}{s'}{v'}{s''}\]
      Hence,
      \[\reductionI[\env{\rhot\concat(x,l)}{(x,l')\concat\rho}]{a\ ;\ b}{s}{v'}{s''} \ssi
      \reductionI[\env{\rhot\concat(x,l)}{\rho}]{a\ ;\ b}{s}{v'}{s''}\]
    
 The other cases are similar.

\paragraph{(val)}
      $\reductionI[\env{\rhot\concat(x,l)}{\rho}]{v}{s}{v}{\gc{\rhot\concat(x,l)}{s}} \ssi
      \reductionI[\env{\rhot\concat(x,l)}{(x,l')\concat\rho}]{v}{s}{v}{\gc{\rhot\concat(x,l)}{s}}$

\paragraph{(var)}
      $\rhot\concat(x,l)\concat\rho = \rhot\concat(x,l)\concat(x,l')\concat\rho$
      so, with $l'' = \rhot\concat(x,l)\concat\rho\ y$,
      \[\reductionI[\env{\rhot\concat(x,l)}{\rho}]{y}{s}{s\ l''}{\gc{\rhot\concat(x,l)}{s}}
      \ssi
      \reductionI[\env{\rhot\concat(x,l)}{(x,l')\concat\rho}]{y}{s}{s\ l''}{\gc{\rhot\concat(x,l)}{s}}\]

\paragraph{(assign)}
      $\rhot\concat(x,l)\concat\rho = \rhot\concat(x,l)\concat(x,l')\concat\rho$.
      So,
      \[\reductionI[\env{}{\rhot\concat(x,l)\concat(x,l')\concat\rho}]{a}{s}{v}{s'} \ssi
      \reductionI[\env{}{\rhot\concat(x,l)\concat\rho}]{a}{s}{v}{s'}\]
      Hence, with $l'' = \rhot\concat(x,l)\concat\rho\ y$,
      \[\reductionI[\env{\rhot\concat(x,l)}{\rho}]{y \coloneqq
      a}{s}{\unit}{\gc{\rhot\concat(x,l)}{\subst{s'}{l''}{v}}} \ssi
      \reductionI[\env{\rhot\concat(x,l)}{(x,l')\concat\rho}]{y \coloneqq
      a}{s}{\unit}{\gc{\rhot\concat(x,l)}{\subst{s'}{l''}{v}}}\]

\paragraph{(if-true) and (if-false)} are proved similarly to (seq).

\paragraph{(letrec)}
      $\rhot\concat(x,l)\concat\rho = \rhot\concat(x,l)\concat(x,l')\concat\rho = \rho'$.
      Moreover, by the induction hypotheses,
      \[\reductionIF[\env{\rhot\concat(x,l)}{(x,l')\concat\rho}]{b}{s}{v}{s'}{\mathcal{F'}} \ssi
      \reductionIF[\env{\rhot\concat(x,l)}{\rho}]{b}{s}{v}{s'}{\mathcal{F'}}\]
      Hence,
      \begin{gather*}
      \reductionI[\env{\rhot\concat(x,l)}{(x,l')\concat\rho}]{\letrec{f(\range{x}{1}{n})}{a}{b}}{s}{v}{s'}
      \ssi\\
      \reductionI[\env{\rhot\concat(x,l)}{\rho}]{\letrec{f(\range{x}{1}{n})}{a}{b}}{s}{v}{s'}
      \end{gather*}

\paragraph{(call)}
      $\rhot\concat(x,l)\concat\rho = \rhot\concat(x,l)\concat(x,l')\concat\rho$.
      So,
      \[
      \forall i,
      \reductionI[\env{}{\rhot\concat(x,l)\concat(x,l')\concat\rho}]{a_i}{s_i}{v_i}{s_{i+1}} \ssi
      \reductionI[\env{}{\rhot\concat(x,l)\concat\rho}]{a_i}{s_i}{v_i}{s_{i+1}}\]
      Hence,
      \[\reductionI[\env{\rhot\concat(x,l)}{(x,l')\concat\rho}]{f(\range{a}{1}{n})}{s_{1}}{v}{\gc{\rhot\concat(x,l)}{s'}} \ssi
      \reductionI[\env{\rhot\concat(x,l)}{\rho}]{f(\range{a}{1}{n})}{s_{1}}{v}{\gc{\rhot\concat(x,l)}{s'}}.
      \qedhere
       \]
\end{proof}

Now we can show the required lemmas and prove the equivalence between the
intermediate and optimised reduction rules.

\begin{lemma}[Intermediate implies optimised]\label{lem:IimpliesO}
  \[\text{If }\reductionI{M}{s}{v}{s'}
  \text{ then }\reductionR{M}{s}{v}{s'}.
  \]
\end{lemma}
\begin{proof}
  By induction on the structure of the derivation.
  The interesting cases are (letrec) and (call), where
  compact environments are respectively built and used.

\paragraph{(letrec)}
      By the induction hypotheses,
      \[\reductionF{b}{s}{v}{s'}{\close{\mathcal{F'}}}\]
      Since we defined canonical compact environments so as to match
      exactly the way compact environments are built in the optimised
      reduction rules, the constraints of the (letrec) rule are fulfilled:
      \[\close{\mathcal{F'}}=\ajout{\close{\F}}{f}{\fun{\range{x}{1}{n}}{a}{\rho',\close{\F}}},\]
      hence:
      \[\reductionR{\letrec{f(\range{x}{1}{n})}{a}{b}}{s}{v}{s'}\]

\paragraph{(call)}
      By the induction hypotheses,
      \[\forall i, \reductionR[\env{}{\rhot\concat\rho}]{a_i}{s_i}{v_i}{s_{i+1}}\]
      and
      \[\reductionF[\env{\rho''}{\rho'}]{b}{\ajout{s_{n+1}}{l_i}{v_i}}{v}{s'}%
      {\close{(\ajout{\mathcal{F'}}{f}{\F\,f})}}\]
      Lemma~\ref{lem:intro-in-env2} allows to remove hidden variables,
      which leads to
      \[\reductionF[\env{\rho''}{\rho'_{|\dom(\rho')\setminus\{\range{x}{1}{n}\}}}]{b}{\ajout{s_{n+1}}{l_i}{v_i}}{v}{s'}%
         {\close{(\ajout{\mathcal{F'}}{f}{\F\,f})}}\]
      Besides, 
      \[\close{\F}\ f =
      \fun{\range{x}{1}{n}}{b}{\rho'_{|\dom(\rho')\setminus\{\range{x}{1}{n}\}},\close{\mathcal{F'}}}\]
      and
      \[\close{(\ajout{\mathcal{F'}}{f}{\F\,f})} =
        \ajout{\close{\mathcal{F'}}}{f}{\close{\F}\ f}\]
      Hence
      \[\reductionR{f(\range{a}{1}{n})}{s_{1}}{v}{\gc{\rhot}{s'}}.
      \]

\paragraph{(val)}
      \reductionR{v}{s}{v}{\gc{\rhot}{s}}

\paragraph{(var)}
      \reductionR{x}{s}{s\ l}{\gc{\rhot}{s}}

\paragraph{(assign)}
      By the induction hypotheses,
      \reductionR[\env{}{\rhot\concat\rho}]{a}{s}{v}{s'}.
      Hence,
      \[\reductionR{x \coloneqq a}{s}{\unit}{\gc{\rhot}{\subst{s'}{l}{v}}}\]

\paragraph{(seq)}
      By the induction hypotheses,
      \[\reductionR[\env{}{\rhot\concat\rho}]{a}{s}{v}{s'} \qquad\qquad
      \reductionR{b}{s'}{v'}{s''}\]
      Hence,
      \[\reductionR{a\ ;\ b}{s}{v'}{s''}\]

\paragraph{(if-true) and (if-false)} are proved similarly to (seq).
    \qedhere
\end{proof}

\begin{lemma}[Optimised implies intermediate]\label{lem:OimpliesI}
  \[ \text{If }\reduction{M}{s}{v}{s'}
  \text{ then }\forall \mathcal{G} \text{ such that } \close{\mathcal{G}}=\F,
  \reductionIF{M}{s}{v}{s'}{\mathcal{G}}.\]
\end{lemma}
\begin{proof}
  First note that, since $\close{\mathcal{G}}=\F$, $\F$ is necessarily compact.

  By induction on the structure of the derivation.
  The interesting cases are (letrec) and (call), where
  non-compact environments are respectively built and used.

\paragraph{(letrec)} Let $\mathcal{G} \text{ such as } \close{\mathcal{G}}=\F$.
      Remember that
      $\rho' = \rhot\concat\rho|_{\dom(\rhot\concat\rho)\setminus\{\range{x}{1}{n}\}}$.
      Let
      \[\mathcal{G'}=\ajout{\mathcal{G}}{f}{\fun{\range{x}{1}{n}}{a}{\rhot\concat\rho,\F}}\]
      which leads, since $\F$ is compact ($\close{\F} = \F$), to
      \begin{align*}
        \close{\mathcal{G'}}&=\ajout{\F}{f}{\fun{\range{x}{1}{n}}{a}{\rho',\F}}\\
        &=\mathcal{F'}
      \end{align*}
      By the induction hypotheses,
      \[\reductionIF{b}{s}{v}{s'}{\mathcal{G'}}\]
      Hence,
      \[\reductionIF{\letrec{f(\range{x}{1}{n})}{a}{b}}{s}{v}{s'}{\mathcal{G}}\]

\paragraph{(call)} Let $\mathcal{G} \text{ such as } \close{\mathcal{G}}=\F$.
      By the induction hypotheses,
      \[\forall i, \reductionIF[\env{}{\rhot\concat\rho}]{a_i}{s_i}{v_i}{s_{i+1}}{\mathcal{G}}\]
      Moreover,
      since $\close{\mathcal{G}}\ f=\F\,f$,
        \[\mathcal{G}\ f = \fun{\range{x}{1}{n}}{b}{\drange{x}{l}{i}{j}\rho',\mathcal{G'}}\]
        where $\close{\mathcal{G'}} = \mathcal{F'}$,
        and the $l_i$ are some locations stripped out when compacting
        $\mathcal{G}$ to get $\F$.
      By the induction hypotheses,
      \[\reductionIF[\env{\rho''}{\rho'}]{b}{\ajout{s_{n+1}}{l_i}{v_i}}{v}{s'}%
      {\ajout{\mathcal{G'}}{f}{\mathcal{G}\ f}}\]
      Lemma~\ref{lem:intro-in-env2} leads to
      \[\reductionIF[\env{\rho''}{\drange{x}{l}{i}{j}\rho'}]{b}{\ajout{s_{n+1}}{l_i}{v_i}}{v}{s'}%
      {\ajout{\mathcal{G'}}{f}{\mathcal{G}\ f}}\]
      Hence,
      \[\reductionIF{f(\range{a}{1}{n})}{s_{1}}{v}{\gc{\rhot}{s'}}{\mathcal{G}}.
      \]

\paragraph{(val)} $\forall\mathcal{G} \text{ such as } \close{\mathcal{G}}=\F,
      \reductionIF{v}{s}{v}{s'}{\mathcal{G}}$

\paragraph{(var)} $\forall\mathcal{G} \text{ such as } \close{\mathcal{G}}=\F,
      \reductionIF{x}{s}{s\ l}{\gc{\rhot}{s}}{\mathcal{G}}$

\paragraph{(assign)} Let $\mathcal{G} \text{ such as } \close{\mathcal{G}}=\F$.
      By the induction hypotheses,
      $\reductionIF[\env{}{\rhot\concat\rho}]{a}{s}{v}{s'}{\mathcal{G}}$.
      Hence,
     \[\reductionIF{x \coloneqq a}{s}{\unit}{\gc{\rhot}{\subst{s'}{l}{v}}}{\mathcal{G}}\]

\paragraph{(seq)} Let $\mathcal{G} \text{ such as } \close{\mathcal{G}}=\F$.
      By the induction hypotheses,
      \[\reductionIF[\env{}{\rhot\concat\rho}]{a}{s}{v}{s'}{\mathcal{G}}
      \qquad\qquad
      \reductionIF{b}{s'}{v'}{s''}{\mathcal{G}}\]
      Hence
      \[\reductionIF{a\ ;\ b}{s}{v'}{s''}{\mathcal{G}}\]

\paragraph{(if-true) and (if-false)} are proved similarly to (seq).
      \qedhere
\end{proof}

\subsubsection{Intermediate and naive reduction rules equivalence\label{subsec:second-step}}

In this section, we show that the naive and intermediate reduction rules
are equivalent:
\[
\text{Naive rules}
\xtofrom[\text{Lemma~\ref{lem:IimpliesN}}]{\text{Lemma~\ref{lem:NimpliesI}}}
\text{Intermediate rules}
\]
We must therefore show that it is correct to use minimal stores in the
intermediate reduction rules.  We first define a partial order on stores:
\begin{definition}[Store extension]
  \[ s \revextends s' \ssi s'|_{\dom(s)} = s \qedhere\]
\end{definition}
\begin{property}\label{prop:ext-order}
  Store extension ($\revextends$) is a partial order over stores. The
  following operations preserve this order: $\gc{\rho}{\cdot}$ and
  $\subst{\cdot}{l}{v}$, for some given $\rho$, $l$ and $v$.
\end{property}
\begin{proof}
Immediate when considering the stores as function graphs: $\revextends$ is
the inclusion, $\gc{\rho}{\cdot}$ a relative complement, and
$\subst{\cdot}{l}{v}$ a disjoint union (preceded by $\gc{(l,v')}{\cdot}$
when $l$ is already bound to some $v'$).\qedhere
\end{proof}

Before we prove that using minimal stores is equivalent to using full
stores, we need an alpha-conversion lemma, which allows us to rename
locations in the store, provided the new location does not already appear
in the store or the environments.  It is used when choosing
a fresh location for the (call) rule in proofs by induction.
\begin{lemma}[Alpha-conversion]\label{lem:rename-loc}
  If \reductionI{M}{s}{v}{s'}
  then, for all $l$, for all $l'$ appearing neither in $s$ nor in $\F$
  nor in $\rho\concat\rhot$,
  \[\reductionIF[\env{\rhot[l'/l]}{\rho[l'/l]}]{M}{s[l'/l]}{v}{s'[l'/l]}{\F[l'/l]}.\]
  Moreover, both derivations have the same height.
\end{lemma}
\begin{proof}
  By induction on the height of the derivation.
  For the (call) case, we must ensure that the fresh locations $l_i$ do
  not clash with $l'$.  In case they do, we conclude by applying the
  induction hypotheses twice: first to rename the clashing $l_i$ into a
  fresh $l'_i$, then to rename $l$ into $l'$.

  Two preliminary elementary remarks. First, provided $l'$
  appears neither in $\rho$ or $\rhot$, nor in $s$,
  \[(\gc{\rho}{s})[l'/l] = \gc{(\rho[l'/l])}{(s[l'/l])}\] and
  \[(\rhot\concat\rho)[l'/l] = \rhot[l'/l]\concat\rho[l'/l].\]

  Moreover, if
  $\reductionI{M}{s}{v}{s'}$, then $\dom(s') = \dom(s)\setminus\rhot$
  (straightforward by induction). This leads to:
  $\rhot=\varepsilon \Rightarrow \dom(s') = \dom(s)$.

  By induction on the height of the derivation, because the induction
  hypothesis must be applied twice in the case of the (call) rule.

\paragraph{(call)}
     $\forall i, \dom(s_i) = \dom(s_{i+1})$.
     Thus, $\forall i, l' \notin \dom(s_i)$.
     This leads, by the induction hypotheses, to
     \[\forall
     i,\reductionI[\env{}{(\rhot\concat\rho)[l'/l]}]{a_i}{s_i[l'/l]}{v_i}{s_{i+1}[l'/l]}{\F[l'/l]}\]
     Moreover , $\mathcal{F'}$ is part of $\F$.
     As a result, since $l'$ does not appear in $\F$, it does not appear in
     $\mathcal{F'}$, nor in $\ajout{\mathcal{F'}}{f}{\F\,f}$.
     It does not appear in $\rho'$ either (since $\rho'$ is part of
     $\mathcal{F'}$).
     On the other hand, there might be some $j$ such that $l_j = l'$,
     so $l'$ might appear in $\rho''$.
     In that case, we apply the induction hypotheses a first time to rename $l_j$ in some 
     $l_j' \neq l'$. One can chose $l_j'$ such that it does not appear
     in $s_{n+1}$, $\ajout{\mathcal{F'}}{f}{\F\,f}$ nor in 
     $\rho''\concat\rho$. As a result, $l_j'$ is fresh.
     Since $l_j$ is fresh too, and does not appear in $\dom(s')$
     (because of our preliminary remarks),
     this leads to a mere substitution in $\rho''$:
     \[\reductionIF[\env{\rho''[l'_j/l_j]}{\rho'}]{b}{\ajout{s_{n+1}}{l_i[l'_j/l_j]}{v_i}}{v}{s'}{\ajout{\mathcal{F'}}{f}{\F\,f}}\]
     Once this (potentially) disturbing $l_j$ has been renamed (we
     ignore it in the rest of the proof), we apply the induction hypotheses
     a second time to rename $l$ to $l'$:
     \[\reductionIF[\env{\rho''[l'/l]}{\rho'[l'/l]}]{b}{(\ajout{s_{n+1}}{l_i}{v_i})[l'/l]}{v}{s'[l'/l]}{\ajout{\mathcal{F'}}{f}{\F\,f}}\]
     Now,
     $(\ajout{s_{n+1}}{l_i}{v_i})[l'/l] = \ajout{s_{n+1}[l'/l]}{l_i}{v_i}$.
     Moreover,
     \[\F[l'/l]\ f = \fun{\range{x}{1}{n}}{b}{\rho'[l'/l],\mathcal{F'}[l'/l]}\]
     and 
     \[(\ajout{\mathcal{F'}}{f}{\F\,f})[l'/l] =
     \ajout{\mathcal{F'}[l'/l]}{f}{\F[l'/l]\ f}\]
     Finally, $\rho''[l'/l] = \rho''$.
     Hence:
     \[\reductionIF[\env{\rhot[l'/l]}{\rho[l'/l]}]{f(\range{a}{1}{n})}{s_{1}[l'/l]}{v}{\gc{\rhot[l'/l]}{s'[l'/l]}}{\F[l'/l]}.
     \]

\paragraph{(val)}
     $\reductionIF[\env{\rhot[l'/l]}{\rho[l'/l]}]{v}{s[l'/l]}{v}{\gc{\rhot[l'/l]}{s[l'/l]}}{\F[l'/]}$

\paragraph{(var)} $s[l'/l](\rhot[l'/l]\concat\rho[l'/l]\ x) =
   s(\rhot\concat\rho\ x) = v$ implies
     \[\reductionIF[\env{\rhot[l'/l]}{\rho[l'/l]}]{x}{s[l'/l]}{v}{\gc{\rhot[l'/l]}{s[l'/l]}}{\F[l'/]}\]

\paragraph{(assign)} By the induction hypotheses,
     \[\reductionIF[\env{}{(\rhot\concat\rho)[l'/l]}]{a}{s[l'/l]}{v}{s'[l'/]}{\F[l'/]}\]
     Let $s''=\subst{s'}{\rhot\concat\rho\ x}{v}$. Then,
     \[\subst{s'[l'/l]}{(\rhot\concat\rho)[l'/l]\ x}{v} = s''[l'/l]\]
     Hence
     \[\reductionIF[\env{\rhot[l'/l]}{\rho[l'/l]}]{x \coloneqq a}{s[l'/l]}{\unit}{\gc{\rhot[l'/l]}{s''[l'/l]}}{\F[l'/]}\]

\paragraph{(seq)} By the induction hypotheses,
     \[\reductionIF[\env{}{(\rhot\concat\rho)[l'/l]}]{a}{s[l'/l]}{v}{s'[l'/l]}{\F[l'/]}\]
     Besides, $\dom(s')=\dom(s)$, therefore $l' \notin \dom(s')$.
     Then, by the induction hypotheses,
     \[\reductionIF[\env{\rhot[l'/l]}{\rho[l'/l]}]{b}{s'[l'/l]}{v'}{s''[l'/l]}{\F[l'/]}\]
     Hence
     \[\reductionIF[\env{\rhot[l'/l]}{\rho[l'/l]}]{a\ ;\ b}{s[l'/l]}{v'}{s''[l'/l]}{\F[l'/]}\]

\paragraph{(if-true) and (if-false)} are proved similarly to (seq).

\paragraph{(letrec)}
     Since $l'$ appears neither in $\rho'$ nor in $\F$, it does not
     appear in $\mathcal{F'}$ either.
     By the induction hypotheses,
     \[\reductionIF[\env{\rhot[l'/l]}{\rho[l'/l]}]{b}{s[l'/l]}{v}{s'[l'/l]}{\mathcal{F'}[l'/l]}\]
     Moreover,
     \[\mathcal{F'}[l'/l]=\ajout{\F[l'/l]}{f}{\fun{\range{x}{1}{n}}{a}{\rho'[l'/l],\F}}\]
     Hence
     \[\reductionIF[\env{\rhot[l'/l]}{\rho[l'/l]}]{\letrec{f(\range{x}{1}{n})}{a}{b}}{s}{v}{s'}{\F[l'/]}
   \qedhere
   \]
\end{proof}

To prove that using minimal stores is correct, we need to extend them so as
to recover the full stores of naive reduction.  The following lemma shows
that extending a store before an (intermediate) reduction extends the
resulting store too:
\begin{lemma}[Extending a store in a derivation]
\label{lem:extend-store}
   \[ 
   \text{Given the reduction }\reductionI{M}{s}{v}{s'},
    \text{ then }
   \forall t \extends s, \exists t'\extends s',
   \reductionI{M}{t}{v}{t'}.
    \]
  Moreover, both derivations have the same height.
\end{lemma}
\begin{proof}
  By induction on the height of the derivation.  The most interesting case is
  (call), which requires alpha-converting a location (hence the induction on the
  height rather than the structure of the derivation).
  
(var), (val) and (assign) are straightforward by the induction hypotheses and
Property~\ref{prop:ext-order}; (seq), (if-true), (if-false) and (letrec) are
straightforward by the induction hypotheses.

\paragraph{(call)} Let $t_1 \extends s_1$.
      By the induction hypotheses,
      \begin{align*}
        \exists t_2 \extends s_2&,
        \reductionI[\env{}{\rhot\concat\rho}]{a_1}{t_1}{v_1}{t_2}\\
        \exists t_{i+1} \extends s_{i+1}&,
        \reductionI[\env{}{\rhot\concat\rho}]{a_i}{t_i}{v_i}{t_{i+1}}\\
        \exists t_{n+1} \extends s_{n+1}&,
        \reductionI[\env{}{\rhot\concat\rho}]{a_n}{t_n}{v_n}{t_{n+1}}
      \end{align*}
      The locations $l_i$ might belong to $\dom(t_{n+1})$ and thus not
      be fresh.
      By alpha-conversion (Lemma~\ref{lem:rename-loc}), 
      we chose fresh $l'_i$ (not in $\Image(\rho')$ and $\dom(s')$) such that
      \[\reductionIF[\env{(l'_i,v_i)}{\rho'}]{b}{\ajout{s_{n+1}}{l'_i}{v_i}}{v}{s'}%
      {\ajout{\mathcal{F'}}{f}{\F\,f}}\]
      By Property~\ref{prop:ext-order},
      $\ajout{t_{n+1}}{l'_i}{v_i} \extends \ajout{s_{n+1}}{l'_i}{v_i}$.
      By the induction hypotheses,
      \[\exists t' \extends s',
      \reductionIF[\env{(l'_i,v_i)}{\rho'}]{b}{\ajout{t_{n+1}}{l'_i}{v_i}}{v}{t'}%
      {\ajout{\mathcal{F'}}{f}{\F\,f}}\]
      Moreover,
      $\gc{\rhot}{t'} \extends \gc{\rhot}{s'}$.
      Hence,
      \[\reductionI{f(\range{a}{1}{n})}{t_{1}}{v}{\gc{\rhot}{t'}}.
      \]

\paragraph{(var)} Let $t \extends s$.
      $\reductionI{v}{t}{v}{\gc{\rhot}{t}}$ and
      $\exists t' = \gc{\rhot}{t} \extends \gc{\rhot}{s} = s'$
      (Property~\ref{prop:ext-order}).

\paragraph{(val)} Let $t \extends s$.
      $\reductionI{x}{t}{t\ l}{\gc{\rhot}{t}}$ and
      $\exists t' = \gc{\rhot}{t} \extends \gc{\rhot}{s} = s'$ 
      (Property~\ref{prop:ext-order}).
      Moreover,
      $t\ l = s\ l$
      because
      $l \in \dom(s)$ and $t|_{\dom(s)} = s$.

\paragraph{(assign)} Let $t \extends s$.
      By the induction hypotheses,
      \[\exists t' \extends s',\reductionI[\env{}{\rhot\concat\rho}]{a}{t}{v}{t'}\]
      Hence,
      \[\reductionI{x \coloneqq a}{t}{\unit}{\gc{\rhot}{\subst{t'}{l}{v}}}\]
      concludes, since
      $\gc{\rhot}{\subst{t'}{l}{v}} \extends \gc{\rhot}{\subst{t'}{l}{v}}$
      (Property~\ref{prop:ext-order}).

\paragraph{(seq)} Let $t \extends s$.
      By the induction hypotheses,
      \begin{align*}
        \exists t' \extends s'&, \reductionI[\env{}{\rhot\concat\rho}]{a}{t}{v}{t'}\\
        \exists t''\extends s''&, \reductionI{b}{t'}{v'}{t''}
      \end{align*}
      Hence,
      \[\exists t''\extends s'',\reductionI{a\ ;\ b}{t}{v'}{t''}\]

\paragraph{(if-true) and (if-false)} are proved similarly to (seq).

\paragraph{(letrec)} Let $t \extends s$.
      By the induction hypotheses,
      \[\exists t'\extends s',\reductionIF{b}{s}{v}{s'}{\mathcal{F'}}\]
      Hence,
      \[\exists t'\extends
      s',\reductionI{\letrec{f(\range{x}{1}{n})}{a}{b}}{s}{v}{t'}
      \qedhere
  \]
\end{proof}

Now we can show the required lemmas and prove the equivalence between the
intermediate and naive reduction rules.

\begin{lemma}[Intermediate implies naive]\label{lem:IimpliesN}
  \[\text{If }\reductionI{M}{s}{v}{s'}
    \text{ then }
    \exists t'\extends s', \reductionN[\rhot\concat\rho]{M}{s}{v}{t'}.
  \]
\end{lemma}
\begin{proof}
  By induction on the height of the derivation, because some stores are
  modified during the proof.
  The interesting cases are (seq) and (call), where
  Lemma~\ref{lem:extend-store} is used to extend intermediary stores.
  Other cases are straightforward by Property~\ref{prop:ext-order} and
  the induction hypotheses.

\paragraph{(seq)}
      By the induction hypotheses,
      \[\exists t' \extends s', \reductionN{a}{s}{v}{t'}.\]
      Moreover,
      \[\reductionI{b}{s'}{v'}{s''}.\]
      Since $t' \extends s'$, Lemma~\ref{lem:extend-store} leads to:
      \[\exists t\extends s'', \reductionI{b}{t'}{v'}{t}\]
      and the height of the derivation is preserved.
      By the induction hypotheses,
      \[\exists t''\extends t, \reductionN{b}{t'}{v'}{t''}\]
      Hence,
      since $\revextends$ is transitive
      (Property~\ref{prop:ext-order}),
      \[\exists t''\extends s'',\reductionN{a\ ;\ b}{s}{v'}{t''}.\]

\paragraph{(call)}
      Similarly to the (seq) case, we apply the induction
      hypotheses and Lemma~\ref{lem:extend-store}:
      \begin{align*}
        \exists t_2 \extends s_2&,
        \reductionN{a_1}{s_1}{v_1}{t_2}&\text{(Induction)}\\
        \exists t'_{i+1} \extends s_{i+1}&,
        \reductionI[\env{}{\rhot\concat\rho}]{a_i}{t_i}{v_i}{t'_{i+1}}&\text{(Lemma~\ref{lem:extend-store})}\\
        \exists t_{i+1} \extends t'_{i+1} \extends s_{i+1}&,
        \reductionN{a_i}{t_i}{v_i}{t_{i+1}}&\text{(Induction)}\\
        \exists t'_{n+1} \extends s_{n+1}&,
        \reductionI[\env{}{\rhot\concat\rho}]{a_n}{t_n}{v_n}{t'_{n+1}}&\text{(Lemma~\ref{lem:extend-store})}\\
        \exists t_{n+1} \extends t'_{n+1} \extends s_{n+1}&,
        \reductionN{a_n}{t_n}{v_n}{t_{n+1}}&\text{(Induction)}
      \end{align*}
      The locations $l_i$ might belong to $\dom(t_{n+1})$ and thus not
      be fresh.
      By alpha-conversion (Lemma~\ref{lem:rename-loc}), 
      we choose a set of fresh $l'_i$ (not in $\Image(\rho')$ and $\dom(s')$) such that
      \[\reductionIF[\env{(l'_i,v_i)}{\rho'}]{b}{\ajout{s_{n+1}}{l'_i}{v_i}}{v}{s'}%
      {\ajout{\mathcal{F'}}{f}{\F\,f}}.\]
      By Property~\ref{prop:ext-order},
      $\ajout{t_{n+1}}{l'_i}{v_i} \extends \ajout{s_{n+1}}{l'_i}{v_i}$.
      Lemma~\ref{lem:extend-store} leads to,
      \[\exists t \extends s',
      \reductionIF[\env{(l'_i,v_i)}{\rho'}]{b}{\ajout{t_{n+1}}{l'_i}{v_i}}{v}{t}%
      {\ajout{\mathcal{F'}}{f}{\F\,f}}.\]
      By the induction hypotheses,
      \[\exists t' \extends t \extends s',
      \reductionNF[(l'_i,v_i)\concat\rho']{b}{\ajout{t_{n+1}}{l'_i}{v_i}}{v}{t'}%
      {\ajout{\mathcal{F'}}{f}{\F\,f}}.\]
      Moreover,
      $\gc{\rhot}{t'} \extends \gc{\rhot}{s'}$.
      Hence,
      \[\reductionN{f(\range{a}{1}{n})}{s_1}{v}{\gc{\rhot}{t'}}.
      \]

\paragraph{(val)} $\reductionN{v}{s}{v}{t'}$ with $t' = s \extends
      \gc{\rhot}{s} = s'$.

\paragraph{(var)} $\reductionN{x}{s}{s\ l}{s''}$ with $t' = s \extends
      \gc{\rhot}{s} = s'$.

\paragraph{(assign)}
      By the induction hypotheses,
      \[\exists s''\extends s',\reductionN{a}{s}{v}{t'}\]
      Hence,
      \[\reductionN{x \coloneqq a}{s}{\unit}{\subst{t'}{l}{v}}\]
      concludes since
      $\subst{t'}{l}{v}\extends\subst{s'}{l}{v}$
      (Property~\ref{prop:ext-order}).

\paragraph{(if-true) and (if-false)} are proved similarly to (seq).

\paragraph{(letrec)}
      By the induction hypotheses,
      \[\exists t'\extends s',\reductionNF{b}{s}{v}{s'}{\mathcal{F'}}.\]
      Hence,
      \[\exists t'\extends
      s',\reductionN{\letrec{f(\range{x}{1}{n})}{a}{b}}{s}{v}{t'}.
      \qedhere
  \]
\end{proof}

The proof of the converse property --- i.e.\  if a term reduces in the
naive reduction rules, it reduces in the intermediate reduction rules
too --- is more complex because the naive reduction rules provide very weak
invariants about stores and environments.  For that reason, we add an
hypothesis to ensure that every location appearing in the environments
$\rho$, $\rhot$ and $\F$ also appears in the store $s$:
  \[\Image(\rhot\concat\rho)\cup\Loc(\F) \subset \dom(s).\]
Moreover, since stores are often larger in the naive reduction rules
than in the intermediate ones, we need to generalise the induction
hypothesis.

\begin{lemma}[Naive implies intermediate]\label{lem:NimpliesI}

Assume
  $\Image(\rhot\concat\rho)\cup\Loc(\F) \subset \dom(s)$.
  Then,
  $\reductionN[\rhot\concat\rho]{M}{s}{v}{s'}$
  implies
  \[
  \forall t \revextends s
  \text{ such that }
  \Image(\rhot\concat\rho)\cup\Loc(\F) \subset \dom(t),
  \quad
  \reductionI{M}{t}{v}{s'|_{\dom(t)\setminus\Image(\rhot)}}.
  \]
\end{lemma}
\begin{proof}
  By induction on the structure of the derivation.

\paragraph{(val)}  Let $t \revextends s$. Then
     \begin{align*}
       \gc{\rhot}{t}
       &= s|_{\dom(t)\setminus\Image(\rhot)}
       && \text{because $s|_{\dom(t)} = t$}\\
       &= s'|_{\dom(t)\setminus\Image(\rhot)}
       && \text{because $s' = s$}
     \end{align*}
     Hence, \[\reductionI{v}{t}{v}{\gc{\rhot}{t}}.\]

\paragraph{(var)} Let $t \revextends s$ such that
     $\Image(\rhot\concat\rho)\cup\Loc(\F) \subset \dom(t)$.
     Note that
     $l \in \Image(\rhot\concat\rho) \subset \dom(t)$
     implies
     $t\ l = s\ l$.
     Then,
     \begin{align*}
       \gc{\rhot}{t}
       &= s|_{\dom(t)\setminus\Image(\rhot)}
       && \text{because $s|_{\dom(t)} = t$}\\
       &= s'|_{\dom(t)\setminus\Image(\rhot)}
       && \text{because $s' = s$}
     \end{align*}
     Hence, \[\reductionI{x}{t}{t\ l}{\gc{\rhot}{t}}.\]

\paragraph{(assign)} Let $t \revextends s$ such that
     $\Image(\rhot\concat\rho)\cup\Loc(\F) \subset \dom(t)$.
     By the induction hypotheses,
     since $\Image(\varepsilon) = \emptyset$,
     \[\reductionI[\env{}{\rhot\concat\rho}]{a}{t}{v}{s'|_{\dom(t)}}\]
     Note that
     $l \in \Image(\rhot\concat\rho) \subset \dom(t)$
     implies
     $l \in \dom(s'|_{\dom(t)})$.
     Then
     \begin{align*}
       \gc{\rhot}{(\subst{s'|_{\dom(t)}}{l}{v})}
       &=\gc{\rhot}{(\subst{s'}{l}{v})|_{\dom(t)}}
         &\text{because $l \in \dom(s'|_{\dom(t)})$}\\
       &=(\subst{s'}{l}{v})|_{\dom(t)\setminus\Image(\rhot)}
     \end{align*}
     Hence, \[ \reductionI{x \coloneqq a}{s}{\unit}{\gc{\rhot}%
     {(\subst{s'|_{\dom(t)}}{l}{v})}}. \]

\paragraph{(seq)} Let $t \revextends s$ such that
     $\Image(\rhot\concat\rho)\cup\Loc(\F) \subset \dom(t)$.
     By the induction hypotheses,
     since $\Image(\varepsilon) = \emptyset$,
     \[\reductionI[\env{}{\rhot\concat\rho}]{a}{t}{v}{s'|_{\dom(t)}}\]
     Moreover,
     $s'|_{\dom(t)} \revextends s'$
     and
     $\Image(\rhot\concat\rho)\cup\Loc(\F) \subset \dom(s'|_{\dom(t)})=\dom(t)$.
     By the induction hypotheses, this leads to:
     \[\reductionI{b}{s'|_{\dom(t)}}{v'}{s''|_{\dom(s'|_{\dom(t)})\setminus\Image(\rhot)}}.\]
     Hence,
     with $\dom(s'|_{\dom(t)})=\dom(t)$,
     \[\reductionI{a\ ;\ b}{t}{v'}{s''|_{\dom(t)\setminus\Image(\rhot)}}.\]

\paragraph{(if-true) and (if-false)} are proved similarly to (seq).

\paragraph{(letrec)} Let $t \revextends s$ such that
     $\Image(\rhot\concat\rho)\cup\Loc(\F) \subset \dom(t)$.
     \[\Loc(\mathcal{F'}) = \Loc(\F)\cup\Image(\rhot\concat\rho)
     \text{ implies }
     \Image(\rhot\concat\rho)\cup\Loc(\mathcal{F'}) \subset \dom(t).\]
     Then, by the induction hypotheses,
     \[\reductionIF{b}{t}{v}{s'|_{\dom(t)\setminus\Image(\rhot)}}{\mathcal{F'}}.\]
     Hence,
     \[\reductionI{\letrec{f(\range{x}{1}{n})}{a}{b}}{t}{v}{s'|_{\dom(t)\setminus\Image(\rhot)}}.\]

\paragraph{(call)} Let $t \revextends s_1$ such that
     $\Image(\rhot\concat\rho)\cup\Loc(\F) \subset \dom(t)$.
     Note the following equalities:
      \begin{align*}
      s_1|_{\dom(t)} &= t\\
      s_2|_{\dom(t)} &\revextends s_2\\
      \Image(\rhot\concat\rho)\cup\Loc(\F) \subset \dom(s_2|_{\dom(t)})
      & = \dom(t)\\
      s_3|_{\dom(s_2|_{\dom(t)})} &= s_3|_{\dom(t)}
      \end{align*}
     By the induction hypotheses, they yield:
     \begin{align*}
       \reductionI[\env{}{\rhot\concat\rho}]{a_1}{t}{v_1}{s_2|_{\dom(t)}}\\
       \reductionI[\env{}{\rhot\concat\rho}]{a_2}{s_2|_{\dom(t)}}{v_1}{s_3|_{\dom(t)}}\\
       \forall i,\reductionI[\env{}{\rhot\concat\rho}]{a_i}{s_i|_{\dom(t)}}{v_i}{s_{i+1}|_{\dom(t)}}
     \end{align*}
     Moreover,
     $s_{n+1}|_{\dom(t)} \revextends s_{n+1}$
     implies
     $\ajout{s_{n+1}|_{\dom(t)}}{l_i}{v_i} \revextends \ajout{s_{n+1}}{l_i}{v_i}$ 
     (Property~\ref{prop:ext-order})
     and:
     \begin{align*}
       \Image(\rho''\concat\rho')\cup\Loc(\ajout{\mathcal{F'}}{f}{\F\,f})
       &= \Image(\rho'') \cup (\Image(\rho')\cup\Loc(\mathcal{F'}))\\
       &\subset \{l_i\} \cup \Loc(\F)\\
       &\subset \{l_i\} \cup \dom(t)\\
       &\subset \dom(\ajout{s_{n+1}|_{\dom(t)}}{l_i}{v_i})
     \end{align*}
     Then, by the induction hypotheses,
     \[\reductionIF[\env{\rho''}{\rho'}]{b}{\ajout{s_{n+1}|_{\dom(t)}}{l_i}{v_i}}{v}%
     {s'|_{\dom(\ajout{s_{n+1}|_{\dom(t)}}{l_i}{v_i})\setminus\Image(\rho'')}}%
     {\ajout{\mathcal{F'}}{f}{\F\,f}}\]
     Finally,
     \begin{align*}
       \gc{\rhot}{s'|_{\dom(\ajout{s_{n+1}|_{\dom(t)}}{l_i}{v_i})\setminus\Image(\rho'')}}
       &= \gc{\rhot}{s'|_{\dom(t)\cup\{l_i\}\setminus\{l_i\}}} = \gc{\rhot}{s'|_{\dom(t)}}\\
       &= (\gc{\rhot}{s'})|_{\dom(t)\setminus\Image(\rhot)}
       \quad\text{(by definition of $\gc{\cdot}{\cdot}$)}
     \end{align*}
     Hence,
     \[\reductionI{f(\range{a}{1}{n})}{t}{v}%
     {(\gc{\rhot}{s'})|_{\dom(t)\setminus\Image(\rhot)}}.
     \qedhere
     \]
\end{proof}

\subsection{Correctness of lambda-lifting\label{sec:correction-ll}}

In this section, we prove the correctness of lambda-lifting
(Theorem~\ref{thm:lambda-lifting-correctness},
p.~\pageref{thm:lambda-lifting-correctness}) by induction on the height of the
optimised reduction.

Section~\ref{sec:strong-invariants} defines stronger invariants and
rewords the correctness theorem with them.  Section~\ref{sec:overview}
gives an overview of the proof.  Sections~\ref{sec:rewriting-lemmas}
and~\ref{sec:aliasing-lemmas} prove a few lemmas needed for the proof.
Section~\ref{sec:proof-correctness} contains the actual proof of
correctness.

\subsubsection{Strengthened hypotheses}
\label{sec:strong-invariants}

We need strong induction hypotheses to ensure that key invariants about
stores and environments hold at every step.  For that purpose, we define
\emph{aliasing-free environments}, in which locations may not be
referenced by more than one variable, and \emph{local positions}.  They
yield a strengthened version of liftable parameters
(Definition~\ref{dfn:var-liftable}).  We then define lifted environments
(Definition~\ref{dfn:lifted-env}) to mirror the effect of
lambda-lifting in lifted terms captured in closures, and finally
reformulate the correctness of lambda-lifting in
Theorem~\ref{thm:correction-ll} with hypotheses strong enough to be
provable directly by induction.

\begin{definition}[Aliasing]\label{dfn:aliasing}
A set of environments $\mathcal{E}$ is \emph{aliasing-free}
when:
\[\forall \rho,\rho' \in \mathcal{E}, \forall x \in \dom(\rho), \forall y
\in \dom(\rho'),\
\rho\ x = \rho'\ y \Rightarrow x = y.
\]
By extension, an environment of functions $\F$ is aliasing-free when
$\Env(\F)$ is aliasing-free.
\end{definition}
The notion of aliasing-free environments is not an artifact of our small
language, but translates a fundamental property of the C semantics:
distinct function parameters or local variables are always bound to
distinct memory locations (Section~6.2.2, paragraph~6 in ISO/IEC 9899
\cite{iso9899}).  

A local position is any position in a term except inner functions.
Local positions are used to distinguish functions defined directly in a
term from deeper nested functions, because we need to enforce
Invariant~\ref{case:loc} (Definition~\ref{dfn:var-liftable}) on the
former only.
\begin{definition}[Local position]
\emph{Local positions} are defined inductively as follows:
\begin{enumerate}
\item $M$ is in local position in $M$, $x \coloneqq M$, $M \ ;\ M$,
  \ite{M}{M}{M} and $f(M,\dotsc,M)$.
\item $N$ is in local position in \letrec{f(\range{x}{1}{n})}{M}{N}.
    \qedhere
\end{enumerate}
\end{definition}

We extend the notion of liftable parameter
(Definition~\ref{dfn:var-liftable-simple},
p.~\pageref{dfn:var-liftable-simple}) to enforce invariants on stores and
environments.
\begin{definition}[Extended liftability]\label{dfn:var-liftable}
  The parameter $x$ is \emph{liftable} in $(M,\F,\rhot,\rho)$ when:
  \begin{enumerate}
    \item $x$ is defined as the parameter of a function $g$,
    either in $M$ or in $\F$,
    \label{case:def}
    \item in both $M$ and $\F$,
    inner functions in $g$, named $h_i$, are defined and called
    exclusively:
      \begin{enumerate}
      \item in tail position in $g$, or
      \item in tail position in some $h_j$ (with possibly $i=j$), or
      \item in tail position in $M$,
      \end{enumerate}
    \label{case:pos}
    \item for all $f$ defined in local position in $M$,
    $x \in \dom(\rhot\concat\rho) \Leftrightarrow \exists i, f = h_i$,
    \label{case:loc}
    \item moreover,
    if $h_i$ is called in tail position in $M$,
    then $x \in \dom(\rhot)$,
    \label{case:term}
    \item in \F,
    $x$ appears necessarily and exclusively in the environments of the
    $h_i$'s closures,
    \label{case:exclu}
    \item $\F$ contains only compact closures and
    $\Env(\F)\cup\{\rho,\rhot\}$ is aliasing-free.
    \label{case:share}
    \qedhere
  \end{enumerate}
\end{definition}

We also extend the definition of lambda-lifting
(Definition~\ref{dfn:lifted-term}, p.~\pageref{dfn:lifted-term}) to
environments, in order to reflect changes in lambda-lifted parameters
captured in closures.
\begin{definition}[Lifted form of an environment]\label{dfn:lifted-env}
  \begin{align*}
  \text{If } \F\,f =&
  \fun{\range{x}{1}{n}}{b}{\rho',\mathcal{F'}}\qquad\text{then}\\
  \lift{\F}\ f=&
  \begin{cases}
  \fun{\range{x}{1}{n}x}{\lift{b}}{\rho'|_{\dom(\rho')\setminus\{x\}},\lift{\mathcal{F'}}}&\text{when
  $f = h_i$ for some $i$}\\
  \fun{\range{x}{1}{n}}{\lift{b}}{\rho',\lift{\mathcal{F'}}}&\text{otherwise}
  \qedhere
  \end{cases}
  \end{align*}
\end{definition}
Lifted environments are defined such that a liftable parameter never
appears in them.  This property will be useful during the proof of
correctness.
\begin{lemma}\label{lem:Fstarclean}
  If $x$ is a liftable parameter in $(M,\F,\rhot,\rho)$,
  then $x$ does not appear in \lift{\F}.
\end{lemma}
\begin{proof}
  Since $x$ is liftable in $(M, \F, \rhot, \rho)$,
  it appears exclusively in the environments of $h_i$.
  By definition, it is removed when building \lift{\F}. \qedhere
\end{proof}

These invariants and definitions lead to a correctness theorem with stronger
hypotheses.
\begin{theorem}[Correctness of lambda-lifting]\label{thm:correction-ll}
If $x$ is a liftable parameter in $(M,\F,\rhot,\rho)$, then
\[\reduction{M}{s}{v}{s'} \text{ implies }
\reductionF{\lift{M}}{s}{v}{s'}{\lift{\F}}\]
\end{theorem}
Since naive and optimised reductions rules are equivalent
(Theorem~\ref{thm:sem-equiv}, p.~\pageref{thm:sem-equiv}), the proof of
Theorem~\ref{thm:lambda-lifting-correctness}
(p.~\pageref{thm:lambda-lifting-correctness}) is a direct corollary of this
theorem.
\begin{corollary}
If $x$ is a liftable parameter in $M$, then
\[\exists t,
\reductionNF[\varepsilon]{M}{\varepsilon}{v}{t}{\varepsilon} \text{ implies }
\exists t',
\reductionNF[\varepsilon]{\lift{M}}{\varepsilon}{v}{t'}{\varepsilon}.\]
\end{corollary}

\subsubsection{Overview of the proof}
\label{sec:overview}

With the enhanced liftability definition, we have invariants strong enough
to perform a proof by induction of the correctness theorem.
This proof is detailed in Section~\ref{sec:proof-correctness}.

The proof is not by structural induction but by induction on the height
of the derivation.  This is necessary because, even with the stronger
invariants, we cannot apply the induction hypotheses directly to the
premises in the case of the (call) rule: we have to change the stores
and environments, which means rewriting the whole derivation tree,
before using the induction hypotheses.

To deal with this most difficult case, we distinguish between calling one
of the lifted functions ($f = h_i$) and calling another function (either
$g$, where $x$ is defined, or any other function outside of $g$).  Only the
former requires rewriting; the latter follows directly from the induction
hypotheses.

In the (call) rule with $f = h_i$, issues arise when reducing the body
$b$ of the lifted function.  During this reduction, indeed, the store
contains a new location $l'$ bound by the environment to the lifted
variable $x$, but also contains the location $l$ which
contains the original value of $x$.  Our goal is to show that the
reduction of $b$ implies the reduction of $\lift{b}$, with store and
environments fulfilling the constraints of the (call) rule.

To obtain the reduction of the lifted body $\lift{b}$, we modify the
reduction of $b$ in a series of steps, using several lemmas:
\begin{itemize}
    \item the location $l$ of the free variable $x$ is moved to the tail
        environment (Lemma~\ref{lem:switch-x});
    \item the resulting reduction meets the induction hypotheses, which
        we apply to obtain the reduction of the lifted body $\lift{b}$;
    \item however, this reduction does not meet the constraints of the
        optimised reduction rules because the location $l$ is not fresh:
        we rename it to a fresh location $l'$ to hold the lifted
        variable (Lemma~\ref{lem:rename-loc-opt});
    \item finally, since we renamed $l$ to $l'$, we need to reintroduce
        a location $l$ to hold the original value of $x$
        (Lemmas~\ref{lem:intro-in-store} and~\ref{lem:intro-in-env}).
\end{itemize}
The rewriting lemmas used in the (call) case are shown in
Section~\ref{sec:rewriting-lemmas}.

For every other case, the proof consists in checking thoroughly that the
induction hypotheses apply, in particular that $x$ is liftable in the
premises.  These verifications consist in checking Invariants~\ref{case:loc}
to~\ref{case:share} of the extended liftability definition
(Definition~\ref{dfn:var-liftable}) --- Invariants~\ref{case:def}
and~\ref{case:pos} are obvious enough not to be detailed.  To keep the main
proof as compact as possible, the most difficult cases of liftability, related
to aliasing, are proven in some preliminary lemmas
(Section~\ref{sec:aliasing-lemmas}).

One last issue arises during the induction when one of the premises does
not contain the lifted variable $x$.  In that case, the invariants do
not hold, since they assume the presence of $x$.  But it turns out that
in this very case, the lifting function is the identity (since there is
no variable to lift) and lambda-lifting is trivially correct.

\subsubsection{Rewriting lemmas} \label{sec:rewriting-lemmas}
Calling a lifted function has an impact on the resulting store: new
locations are introduced for the lifted parameters and the earlier
locations, which are not modified anymore, are hidden.  Because of these
changes, the induction hypotheses do not apply directly in the case of
the (call) rule for a lifted function $h_i$. We use the following four
lemmas to obtain, through several rewriting steps, a reduction of lifted
terms meeting the induction hypotheses.

\begin{itemize}
    \item Lemma~\ref{lem:switch-x} shows that moving a variable from the non-tail
environment $\rho$ to the tail environment $\rhot$ does not change the
result, but restricts the domain of the store.  It is used transform the
original free variable $x$ (in the non-tail environment) to its lifted
copy (which is a parameter of $h_i$, hence in the tail environment).
\item Lemma~\ref{lem:rename-loc-opt} handles alpha-conversion in stores
    and is used when choosing a fresh location.
\item Lemmas~\ref{lem:intro-in-store} and~\ref{lem:intro-in-env} finally
    add into the store and the environment a fresh location, bound to an
    arbitrary value.  It is used to reintroduce the location containing
    the original value of $x$, after it has been alpha-converted to
    $l'$.
\end{itemize}

\begin{lemma}[Switching to tail environment]\label{lem:switch-x}
  If $\reduction[\env{\rhot}{(x,l)\concat\rho}]{M}{s}{v}{s'}$ and $x \notin
  \dom(\rhot)$ then
  $\reduction[\env{\rhot\concat(x,l)}{\rho}]{M}{s}{v}{s'|_{\dom(s')\setminus\{l\}}}$.
  Moreover, both derivations have the same height.
\end{lemma}
\begin{proof}
  By induction on the structure of the derivation.
  For the (val), (var), (assign) and (call) cases, we use the fact that
  $\gc{\rhot\concat(x,l)}{s} = s'|_{\dom(s')\setminus\{l\}}$ when
  $s' = \gc{\rhot}{s}$.

\paragraph{(val)}
      $\reduction[\env{\rhot\concat(x,l)}{\rho}]{v}{s}{v}{\gc{\rhot\concat(x,l)}{s}}$ and
      $\gc{\rhot\concat(x,l)}{s} = s'|_{\dom(s')\setminus\{l\}}$ with
      $s' = \gc{\rhot}{s}$.

\paragraph{(var)}
      $\reduction[\env{\rhot\concat(x,l)}{\rho}]{y}{s}{s\ l'}{\gc{\rhot\concat(x,l)}{s}}$ and
      $\gc{\rhot\concat(x,l)}{s} = s'|_{\dom(s')\setminus\{l\}}$, with
      $l' = \rhot\concat(x,l)\concat\rho\ y$ and $s' = \gc{\rhot}{s}$.

\paragraph{(assign)}
    By hypothesis,
      $\reduction[\env{}{\rhot\concat(x,l)\concat\rho}]{a}{s}{v}{s'}$
    hence
      $\reduction[\env{\rhot\concat(x,l)}{\rho}]{y \coloneqq a}{s}{\unit}{\gc{\rhot\concat(x,l)}{\subst{s'}{l'}{v}}}$
      and $\gc{\rhot\concat(x,l)}{\subst{s'}{l'}{v}} = s'|_{\dom(s')\setminus\{l\}}$
      with $l' = \rhot\concat(x,l)\concat\rho\ y$ and $s' = \gc{\rhot}{\subst{s'}{l'}{v}}$.

\paragraph{(seq)}
    By hypothesis,
     $\reduction[\env{}{\rhot\concat(x,l)\concat\rho}]{a}{s}{v}{s'}$
     and, by the induction hypotheses,
     $\reduction[\env{\rhot\concat(x,l)}{\rho}]{b}{s'}{v}{s''|_{\dom(s'')\setminus\{l\}}}$
     hence
     \[\reduction[\env{\rhot\concat(x,l)}{\rho}]{a\ ;\ b}{s}{v}{s''|_{\dom(s'')\setminus\{l\}}}.\]

\paragraph{(if-true) and (if-false)} are proved similarly to (seq).

\paragraph{(letrec)}
     By the induction hypotheses,
     \[\reductionF[\env{\rhot\concat(x,l)}{\rho}]{b}{s}{v}{s'|_{\dom(s')\setminus\{l\}}}{\mathcal{F'}}\]
      hence
     \[\reduction[\env{\rhot\concat(x,l)}{\rho}]{\letrec{f(\range{x}{1}{n})}{a}{b}}{s}{v}{s'|_{\dom(s')\setminus\{l\}}}\]

\paragraph{(call)}
      The hypotheses do not change, and the conclusion becomes: 
      \[\reduction[\env{\rhot\concat(x,l)}{\rho}]{f(\range{a}{1}{n})}{s_{1}}{v}{\gc{\rhot\concat(x,l)}{s'}}\]
      as expected, since
      $\gc{\rhot\concat(x,l)}{s'} = s''|_{\dom(s'')\setminus\{l\}}$
      with
      $s'' = \gc{\rhot}{s'}$
  \qedhere
\end{proof}
\begin{lemma}[Alpha-conversion]\label{lem:rename-loc-opt}
  If \reduction{M}{s}{v}{s'}
  then, for all $l$, for all $l'$ appearing neither in $s$ nor in $\F$
  nor in $\rho\concat\rhot$,
  \[\reductionF[\env{\rhot[l'/l]}{\rho[l'/l]}]{M}{s[l'/l]}{v}{s'[l'/l]}{\F[l'/l]}\]
  Moreover, both derivations have the same height.
\end{lemma}
\begin{proof}
    See Lemma~\ref{lem:rename-loc}, p.~\ref{lem:rename-loc}.
\qedhere
\end{proof}
\begin{lemma}[Spurious location in store]\label{lem:intro-in-store}
  If $\reduction{M}{s}{v}{s'}$ and $k$ does not appear in either $s$, $\F$
  or $\rhot\concat\rho$,
  then, for all value $u$,
  $\reduction{M}{\ajout{s}{k}{u}}{v}{\ajout{s'}{k}{u}}$.
  Moreover, both derivations have the same height.
\end{lemma}
\begin{proof}
  By induction on the height of the derivation.  
  The key idea is to add $(k,u)$ to every store in the derivation tree.
  A collision might occur in the (call) rule, if there is some  $j$ such
  that $l_j = k$.  In that case, we need to rename $l_j$ to some fresh
  variable $l'_j \neq k$ (by alpha-conversion) before applying the
  induction hypotheses.

\paragraph{(call)} By the induction hypotheses,
     \[\forall i,
     \reduction[\env{}{\rhot\concat\rho}]{a_i}{\ajout{s_i}{k}{u}}{v_i}{\ajout{s_{i+1}}{k}{u}}\]
     Because $k$ does not appear in $\F$,
     \[k \notin \Loc(\ajout{\mathcal{F'}}{f}{\F\,f}) \subset \Loc(\F)\]
     For the same reason, it does not appear in $\rho'$.
     On the other hand, there might be a $j$ such that $l_j = k$,
     so $k$ might appear in $\rho''$.
     In that case,  we rename $l_j$ in some fresh $l'_j \neq k$,
     appearing in neither $s_{n+1}$, nor $\mathcal{F'}$ or $\rho''\concat\rho'$
     (Lemma~\ref{lem:rename-loc-opt}).
     After this alpha-conversion, $k$ does not appear in either
     $\rho''\concat\rho'$, $\ajout{\mathcal{F'}}{f}{\F\,f}$, or
     $\ajout{s_{n+1}}{l_i}{v_i}$.
     By the induction hypotheses,
     \[\reductionF[\env{\rho''}{\rho'}]{b}{\ajout{\ajout{s_{n+1}}{l_i}{v_i}}{k}{u}}{v}%
     {\ajout{s'}{k}{u}}{\ajout{\mathcal{F'}}{f}{\F\,f}}\]
     Moreover,
     $\gc{\rhot}{\ajout{s'}{k}{u}} = \ajout{\gc{\rhot}{s'}}{k}{u}$
     (since $k$ does not appear in $\rhot$).
     Hence
     \[\reduction{f(\range{a}{1}{n})}{\ajout{s_{1}}{k}{u}}{v}{\gc{\rhot}{\ajout{s'}{k}{u}}}. \]

\paragraph{(val)}
      $\reduction{v}{\ajout{s}{k}{u}}{v}{\gc{\rhot}{\ajout{s}{k}{u}}}$
      and
      $\gc{\rhot}{\ajout{s}{k}{u}} = \ajout{\gc{\rhot}{s}}{k}{u}$
      since $k$ does not appear in $\rhot$.

\paragraph{(var)} $\reduction{x}{\ajout{s}{k}{u}}{(\ajout{s}{k}{u})\
      l}{\gc{\rhot}{\ajout{s}{k}{u}}}$,
      with
      $\gc{\rhot}{\ajout{s}{k}{u}} = \ajout{\gc{\rhot}{s}}{k}{u}$
      since $k$ does not appear in $\rhot$,
      and
      $(\ajout{s}{k}{u})\ l = s\ l$ 
      since
      $k \neq l$ ($k$ does not appear in $s$).

\paragraph{(assign)} By the induction hypotheses,
      $\reduction[\env{}{\rhot\concat\rho}]{a}{\ajout{s}{k}{u}}{v}{\ajout{s'}{k}{u}}$.
      And $k \neq l$ (since $k$ does not appear in $s$)
      then
      $\subst{\ajout{s'}{k}{u}}{l}{v} = \ajout{\subst{s'}{l}{v}}{k}{u}$.
      Moreover, $k$ does not appear in $\rhot$
      then
      $\gc{\rhot}{\ajout{\subst{s'}{l}{v}}{k}{u}} = \ajout{\gc{\rhot}{\subst{s'}{l}{v}}}{k}{u}$.
      Hence
      \[\reduction{x \coloneqq a}{\ajout{s}{k}{u}}{\unit}{\ajout{\gc{\rhot}{\subst{s'}{l}{v}}}{k}{u}}\]

\paragraph{(seq)} By the induction hypotheses,
      \[\reduction[\env{}{\rhot\concat\rho}]{a}{\ajout{s}{k}{u}}{\true}{\ajout{s'}{k}{u}}\]
      \[\reduction{b}{\ajout{s'}{k}{u}}{v'}{\ajout{s''}{k}{u}}\]
      Hence
      \[\reduction{a\ ;\ b}{\ajout{s}{k}{u}}{v'}{\ajout{s''}{k}{u}}\]

\paragraph{(if-true) and (if-false)} are proved similarly to (seq).

\paragraph{(letrec)}
      The location $k$ does not appear in $\mathcal{F'}$,
      because it does not appear in either $\F$ or $\rho'\subset\rhot\concat\rho$
      ($\mathcal{F'}=\ajout{\F}{f}{\fun{\range{x}{1}{n}}{a}{\rho',\F}}$).
      Then, by the induction hypotheses,
      \[\reductionF{b}{\ajout{s}{k}{u}}{v}{\ajout{s'}{k}{u}}{\mathcal{F'}}\]
      Hence
      \[\reduction{\letrec{f(\range{x}{1}{n})}{a}{b}}{\ajout{s}{k}{u}}{v}{\ajout{s'}{k}{u}}.
      \qedhere
  \]
\end{proof}
\begin{lemma}[Spurious variable in environments]\label{lem:intro-in-env}
  \begin{align*}
  \forall l,l', \reduction[\env{\rhot\concat(x,l)}{\rho}]{M}{s}{v}{s'} \ssi&
  \reduction[\env{\rhot\concat(x,l)}{(x,l')\concat\rho}]{M}{s}{v}{s'}
  \end{align*}
  Moreover, both derivations have the same height.
\end{lemma}
\begin{proof}
    See Lemma~\ref{lem:intro-in-env2}, p.~\ref{lem:intro-in-env2}.
       \qedhere
\end{proof}

\subsubsection{Aliasing lemmas} \label{sec:aliasing-lemmas}
We need three lemmas to show that environments remain aliasing-free
during the proof by induction in Section~\ref{sec:proof-correctness}.
The first lemma states that concatenating two environments in an
aliasing-free set yields an aliasing-free set.  The other two prove that
the aliasing invariant (Invariant~\ref{case:share},
Definition~\ref{dfn:var-liftable}) holds in the context of the (call)
and (letrec) rules, respectively.

\begin{lemma}[Concatenation]\label{lem:alias-concat}
    If $\mathcal{E}\cup\{\rho,\rho'\}$ is aliasing-free then
    $\mathcal{E}\cup\{\rho\concat\rho'\}$ is aliasing-free.
\end{lemma}
\begin{proof}
  By exhaustive check of cases.
  We want to prove
  \begin{align*}
  \forall \rho_1,\rho_2 \in \mathcal{E}\cup\{\rho\concat\rho'\},
  \forall x \in \dom(\rho_1),
  \forall y \in \dom(\rho_2),\
  \rho_1\ x = \rho_2\ y \Rightarrow x = y.\\
  \intertext{given that}
  \forall \rho_1,\rho_2 \in \mathcal{E}\cup\{\rho,\rho'\},
  \forall x \in \dom(\rho_1),
  \forall y \in \dom(\rho_2),\
  \rho_1\ x = \rho_2\ y \Rightarrow x = y.
  \end{align*}
  If $\rho_1\in \mathcal{E}$ and $\rho_2\in \mathcal{E}$, immediate.
  If $\rho_1\in \{\rho\concat\rho'\}$,
  $\rho_1\ x = \rho\ x\ \text{or}\ \rho'\ x$. This is the same for
  $\rho_2$. Then $\rho_1\ x = \rho_2\ y$ is equivalent to
   $\rho\ x = \rho'\ y$ (or some other combination, depending on $x$, $y$,
   $\rho_1$ and $\rho_2$) which leads to the expected result.
  \qedhere
\end{proof}

\begin{lemma}[Aliasing in (call) rule]\label{lem:aliasing-call}
Assume that, in a (call) rule,
\begin{itemize}
    \item $\F\,f = \fun{\range{x}{1}{n}}{b}{\rho',\mathcal{F'}}$,
    \item $\Env(\F)$ is aliasing-free, and
    \item $\rho''= \drange{x}{l}{1}{n}$, with fresh and
    distinct locations $l_{i}$.
\end{itemize}
Then
$\Env(\ajout{\mathcal{F'}}{f}{\F\,f})\cup\{\rho',\rho''\}$ is also aliasing-free.
\end{lemma}
\begin{proof}
  Let $\mathcal{E} = \Env(\ajout{\mathcal{F'}}{f}{\F\,f})\cup\{\rho'\}$.
  We know that $\mathcal{E}\subset\Env(\F)$ so $\mathcal{E}$ is aliasing-free
  We want to show that adding fresh and distinct locations from
  $\rho''$ preserves this lack of freedom.  More precisely,
  we want to show that
  \begin{align*}
  \forall \rho_1,\rho_2 \in \mathcal{E}\cup\{\rho''\},
  \forall x \in \dom(\rho_1),
  \forall y \in \dom(\rho_2),\
  \rho_1\ x = \rho_2\ y \Rightarrow x = y\\
  \intertext{given that}
  \forall \rho_1,\rho_2 \in \mathcal{E},
  \forall x \in \dom(\rho_1),
  \forall y \in \dom(\rho_2),\
  \rho_1\ x = \rho_2\ y \Rightarrow x = y.
  \end{align*}
  We reason by checking of all cases.
  If $\rho_1\in \mathcal{E}$ and $\rho_2\in \mathcal{E}$, immediate.
  If $\rho_1 = \rho_2 = \rho''$ then
  $\rho''\ x = \rho''\ y \Rightarrow x = y$
  holds because the locations of $\rho''$ are distinct.
  If $\rho_1 = \rho''$ and
  $\rho_2 \in \mathcal{E}$ then
  $\rho_1\ x = \rho_2\ y \Rightarrow x = y$
  holds because
  $\rho_1\ x \neq \rho_2\ y$
  (by freshness hypothesis). \qedhere
\end{proof}
\begin{lemma}[Aliasing in (letrec) rule]\label{lem:aliasing-letrec}
    If\/ $\Env(\F)\cup\{\rho,\rhot\}$ is aliasing free,
    then, for all $x_i$,
    \[\Env(\F) \cup \{\rho,\rhot\} \cup \{\rhot\concat\rho\
      |_{\dom(\rhot\concat\rho)\setminus\{\range{x}{1}{n}\}}\}\]
    is aliasing free.
\end{lemma}
\begin{proof}
  Let $\mathcal{E} = \Env(\F)\cup\{\rho,\rhot\}$ and
  $\rho'' = \rhot\concat\rho|_{\dom(\rhot\concat\rho)\setminus\{\range{x}{1}{n}\}}$.
  Adding $\rho''$, a restricted concatenation of $\rhot$ and $\rho$, to
  $\mathcal{E}$ preserves aliasing freedom, as in the proof of
  Lemma~\ref{lem:alias-concat}.
  If $\rho_1\in \mathcal{E}$ and $\rho_2\in \mathcal{E}$, immediate.
  If $\rho_1\in \{\rho''\}$,
  $\rho_1\ x = \rho\ x\ \text{or}\ \rho'\ x$. This is the same for
  $\rho_2$. Then $\rho_1\ x = \rho_2\ y$ is equivalent to
   $\rho\ x = \rho'\ y$ (or some other combination, depending on $x$, $y$,
   $\rho_1$ and $\rho_2$) which leads to the expected result. \qedhere
\end{proof}

\subsubsection{Proof of correctness} \label{sec:proof-correctness}

We finally show Theorem~\ref{thm:correction-ll}.
\begin{xrefthm}{thm:correction-ll}
If $x$ is a liftable parameter in $(M,\F,\rhot,\rho)$, then
\[\reduction{M}{s}{v}{s'} \text{ implies }
\reductionF{\lift{M}}{s}{v}{s'}{\lift{\F}}\]
\end{xrefthm}

Assume that $x$ is a liftable parameter in $(M,\F,\rhot,\rho)$. The proof is by
induction on the height of the reduction of $\reduction{M}{s}{v}{s'}$.  To keep
the proof readable, we detail only the non-trivial cases when checking the
invariants of Definition~\ref{dfn:var-liftable} to ensure that the induction
hypotheses hold.

\paragraph{(call) --- first case}
    First, we consider the most interesting case where there exists $i$ such
    that $f = h_i$.
    The variable
    $x$ is a liftable parameter in $(h_i(\range{a}{1}{n}),\F,\rhot,\rho)$
    hence in $(a_i,\F,\varepsilon,\rhot\concat\rho)$ too.

    Indeed, the invariants of Definition~\ref{dfn:var-liftable} hold:
    \begin{itemize}
    \item Invariant~\ref{case:loc}: By definition of a local position, every $f$
    defined in local position in $a_i$ is in local position in
    $h_i(\range{a}{1}{n})$, hence the expected property by the induction hypotheses.
    \item Invariant~\ref{case:term}: Immediate since the premise does not hold : since the
    $a_i$ are not in tail position in $h_i(\range{a}{1}{n})$,
    they cannot feature calls to $h_i$ (by Invariant~\ref{case:pos}).
    \item Invariant~\ref{case:share}: Lemma~\ref{lem:alias-concat},
    p.~\pageref{lem:alias-concat}.
    \end{itemize}
    The other invariants hold trivially.

    By the induction hypotheses, we get
    \[\reduclift[\env{}{\rhot\concat\rho}]{a_{i}}{s_{i}}{v_{i}}{s_{i+1}}.\]
    By definition of lifting, $\lift{h_i(\range{a}{1}{n})} =
    h_i(\lift{a_{1}},\dotsc,\lift{a_{n}},x)$.  But $x$ is not a liftable
    parameter in $(b,\mathcal{F'},\rho'',\rho')$ since the
    Invariant~\ref{case:term} might be broken: $x \notin \dom(\rho'')$
    ($x$ is not a parameter of $h_i$) but $h_j$ might appear in tail
    position in $b$.

    On the other hand, we have $x \in \dom(\rho')$: since, by
    hypothesis, $x$ is a liftable parameter in
    $(h_i(\range{a}{1}{n}),\F,\rhot,\rho)$, it appears necessarily  in
    the environments of the closures of the $h_i$, such as $\rho'$.
    This allows us to split $\rho'$ into two parts:
    $\rho' = (x,l)\concat\rho'''$.
    It is then possible to move $(x,l)$ to the tail environment,
    according to Lemma~\ref{lem:switch-x}:
    \[\reductionF[\env{\rho''(x,l)}{\rho'''}]{b}{\ajout{s_{n+1}}{l_i}{v_i}}{v}{s'|_{\dom(s')\setminus\{l\}}}{\
    \ajout{\mathcal{F'}}{f}{\F\,f}}\]
    This rewriting ensures that
    $x$ is a liftable parameter in $(b,\ajout{\mathcal{F'}}{f}{\F\,f},\rho''\concat(x,l),\rho''')$.

    Indeed, the invariants of Definition~\ref{dfn:var-liftable} hold:
    \begin{itemize}
    \item Invariant~\ref{case:loc}: Every function defined in local position in $b$ is
    an inner function in $h_i$ so, by Invariant~\ref{case:pos},
    it is one of the $h_i$ and
    $x \in \dom(\rho''\concat(x,l)\concat\rho''')$.
    \item Invariant~\ref{case:term}: Immediate since $x \in \dom(\rho''\concat(x,l)\concat\rho''')$.
    \item Invariant~\ref{case:exclu}: Immediate since $\mathcal{F'}$ is included in
    $\F$.
    \item Invariant~\ref{case:share}: Immediate for the compact closures.
    Aliasing freedom is guaranteed by Lemma~\ref{lem:aliasing-call}
    (p.~\pageref{lem:aliasing-call}).
    \end{itemize}
    The other invariants hold trivially.

    By the induction hypotheses,
    \[\reductionF[\env{\rho''(x,l)}{\rho'''}]{\lift{b}}{\ajout{s_{n+1}}{l_i}{v_i}}{v}{s'|_{\dom(s')\setminus\{l\}}}{\
    \lift{\ajout{\mathcal{F'}}{f}{\F\,f}}}\]
    The $l$ location is not fresh: it must be rewritten into a fresh
    location, since $x$ is now a parameter of $h_i$.
    Let $l'$ be a location appearing in neither 
    $\lift{\ajout{\mathcal{F'}}{f}{\F\,f}}$, nor $\ajout{s_{n+1}}{l_i}{v_i}$ or $\rho''\concat\rhot'$.
    Then $l'$ is a fresh location, which is to act as $l$ in the
    reduction of $\lift{b}$.

    We will show that, after the reduction, $l'$ is not in the
    store (just like $l$ before the lambda-lifting).  In the meantime, the value
    associated to $l$ does not change (since $l'$ is modified instead
    of $l$).

    Lemma~\ref{lem:Fstarclean} implies that $x$  does not appear in the
    environments of \lift{\F}, so it does not appear in the environments
    of $\lift{\ajout{\mathcal{F'}}{f}{\F\,f}}\subset\lift{\F}$ either.
    As a consequence, lack of aliasing implies by Definition~\ref{dfn:aliasing}
    that the label $l$, associated to $x$, does not appear in
    $\lift{\ajout{\mathcal{F'}}{f}{\F\,f}}$ either, so
    \[\lift{\ajout{\mathcal{F'}}{f}{\F\,f}}[l'/l] = \lift{\ajout{\mathcal{F'}}{f}{\F\,f}}.\]
    Moreover, $l$ does not appear in $s'|_{\dom(s')\setminus\{l\}}$.
    By alpha-conversion (Lemma~\ref{lem:rename-loc-opt}, since $l'$ does
    not appear in the store or the environments of the reduction, we
    rename $l$ to $l'$:
    \[\reductionF[\env{\rho''(x,l')}{\rho'''}]{\lift{b}}{\ajout{s_{n+1}[l'/l]}{l_i}{v_i}}{v}{s'|_{\dom(s')\setminus\{l\}}}{\
    \lift{\ajout{\mathcal{F'}}{f}{\F\,f}}}.\]
    We want now to reintroduce $l$.  
    Let $v_x = s_{n+1}\ l$.  The location $l$  does not appear in
    $\ajout{s_{n+1}[l'/l]}{l_i}{v_i}$, $\lift{\ajout{\mathcal{F'}}{f}{\F\,f}}$, or $\rho''(x,l')\concat\rho'''$.
    Thus, by Lemma~\ref{lem:intro-in-store},
    \[\reductionF[\env{\rho''(x,l')}{\rho'''}]{\lift{b}}{\ajout{\ajout{s_{n+1}[l'/l]}{l_i}{v_i}}{l}{v_x}}{v}{\ajout{s'|_{\dom(s')\setminus\{l\}}}{l}{v_x}}{\
    \lift{\ajout{\mathcal{F'}}{f}{\F\,f}}}.\]
    Since
    \begin{align*}
      \ajout{\ajout{s_{n+1}[l'/l]}{l_i}{v_i}}{l}{v_x}
      &= \ajout{\ajout{s_{n+1}[l'/l]}{l}{v_x}}{l_i}{v_i}
      && \text{because $\forall i, l \neq l_i$}\\
      &= \ajout{\ajout{s_{n+1}}{l'}{v_x}}{l_i}{v_i}
      && \text{because $v_x = s_{n+1} l$}\\
      &= \ajout{\ajout{s_{n+1}}{l_i}{v_i}}{l'}{v_x}
      && \text{because $\forall i, l' \neq l_i$}
    \end{align*}
    and
    $\ajout{s'|_{\dom(s')\setminus\{l\}}}{l}{v_x} = \ajout{s'}{l}{v_x}$,
    we finish the rewriting by Lemma~\ref{lem:intro-in-env},
    \[\reductionF[\env{\rho''(x,l')}{(x,l)\concat\rho'''}]{\lift{b}}{\ajout{\ajout{s_{n+1}}{l_i}{v_i}}{l'}{v_x}}{v}{\
    \ajout{s'}{l}{v_x}}{\lift{\ajout{\mathcal{F'}}{f}{\F\,f}}}.\]
    Hence the result:
    \[\inferrule*[Left=(call)]{\
    \lift{\F}\ h_i = \fun{\range{x}{1}{n}x}{\lift{b}}{\rho',\lift{\mathcal{F'}}}\\
    \rho''= \drange{x}{l}{1}{n}(x,\rhot\ x)\\
    \text{$l'$ and $l_{i}$ fresh and distinct}\\\\
    \forall i,\reduclift[\env{}{\rhot\concat\rho}]{a_{i}}{s_{i}}{v_{i}}{s_{i+1}} \\
    \reduclift[\env{}{\rhot\concat\rho}]{x}{s_{n+1}}{v_x}{s_{n+1}} \\
    \reductionF[\env{\rho''(x,l')}{\rho'}]{\lift{b}}{\ajout{\ajout{s_{n+1}}{l_i}{v_i}}{l'}{v_x}}{v}{\
    \ajout{s'}{l}{v_x}}{\lift{\ajout{\mathcal{F'}}{f}{\F\,f}}}
    }{\
    \reduclift{h_i(\range{a}{1}{n})}{s_{1}}{v}{\gc{\rhot}{\ajout{s'}{l}{v_x}}}}\]
    Since $l \in \dom(\rhot)$ (because $x$ is a liftable parameter in $(h_i(\range{a}{1}{n}),\F,\rhot,\rho)$),
    the extraneous location is reclaimed as expected:
    $\gc{\rhot}{\ajout{s'}{l}{v_x}} = \gc{\rhot}{s'}$.

\paragraph{(call) --- second case}
    We now consider the case where $f$ is not one of the $h_i$.
    The variable
    $x$ is a liftable parameter in $(f(\range{a}{1}{n}),\F,\rhot,\rho)$
    hence in
    $(a_i,\F,\varepsilon,\rhot\concat\rho)$ too.

    Indeed, the invariants of Definition~\ref{dfn:var-liftable} hold:
    \begin{itemize}
    \item Invariant~\ref{case:loc}: By definition of a local position, every $f$
    defined in local position in $a_i$ is in local position in
    $f(\range{a}{1}{n})$, hence the expected property by the induction hypotheses.
    \item Invariant~\ref{case:term}: Immediate since the premise does not hold : the
    $a_i$ are not in tail position in $f(\range{a}{1}{n})$
    so they cannot feature calls to $h_i$ (by Invariant~\ref{case:pos}:).
    \item Invariant~\ref{case:share}: Lemma~\ref{lem:alias-concat},
    p.~\pageref{lem:alias-concat}.
    \end{itemize}
    The other invariants hold trivially.

    By the induction hypotheses, we get
    \[\reduclift[\env{}{\rhot\concat\rho}]{a_{i}}{s_{i}}{v_{i}}{s_{i+1}},\]
    and, by Definition~\ref{dfn:lifted-term},
    \[\lift{f(\range{a}{1}{n})} = f(\lift{a_{1}},\dotsc,\lift{a_{n}}).\]
    If $x$ is not defined in $b$ or $\F$, then $\lift{}$ is the identity
    function and can trivially be applied to the reduction of $b$.  Otherwise,
    $x$ is a liftable parameter in $(b,\ajout{\mathcal{F'}}{f}{\F\,f},\rho'',\rho')$.

    Indeed, the invariants of Definition~\ref{dfn:var-liftable} hold.
    Assume that $x$ is defined as a parameter of some function $g$, in
    either $b$ or $\F$:
    \begin{itemize}
    \item Invariant~\ref{case:loc}: We have to distinguish the cases where $f = g$
    (with $x \in \dom(\rho'')$)
    and $f \neq g$ (with $x \notin \dom(\rho'')$ and $x \notin \dom(\rho')$).
    In both cases, the result is immediate by the induction hypotheses.
    \item Invariant~\ref{case:term}: If $f \neq g$,
    the premise cannot hold (by the induction hypotheses, Invariant~\ref{case:pos}).
    If  $f = g$, $x \in \dom(\rho'')$ (by the induction hypotheses, Invariant~\ref{case:pos}).
    \item Invariant~\ref{case:exclu}: Immediate since $\mathcal{F'}$ is included in
    $\F$.
    \item Invariant~\ref{case:share}: Immediate for the compact closures.
    Aliasing freedom is guaranteed by Lemma~\ref{lem:aliasing-call}
    (p.~\pageref{lem:aliasing-call}).
    \end{itemize}
    The other invariants hold trivially.

    By the induction hypotheses,
    \[\reductionF[\env{\rho''}{\rho'}]{\lift{b}}{\ajout{s_{n+1}}{l_{i}}{v_{i}}}{v}{s'}{\lift{\ajout{\mathcal{F'}}{f}{\F\,f}}}\]
    hence:
    \[\inferrule*[Left=(call)]{\
    \lift{\F}\ f = \fun{\range{x}{1}{n}}{\lift{b}}{\rho',\lift{\mathcal{F'}}}\\
    \rho''= \drange{x}{l}{1}{n}\\
    \text{$l_{i}$ fresh and distinct}\\\\
    \forall i,\reduclift[\env{}{\rhot\concat\rho}]{a_{i}}{s_{i}}{v_{i}}{s_{i+1}} \\
    \reductionF[\env{\rho''}{\rho'}]{\lift{b}}{\ajout{s_{n+1}}{l_{i}}{v_{i}}}{v}{s'}{\lift{\ajout{\mathcal{F'}}{f}{\F\,f}}}
    }{\
    \reduclift{f(\range{a}{1}{n})}{s_{1}}{v}{\gc{\rhot}{s'}}}\]

\paragraph{(letrec)}
    The parameter $x$ is a liftable in
    $(\letrec{f(\range{x}{1}{n})}{a}{b},\F,\rhot,\rho)$
    so $x$ is a liftable parameter in $(b,\mathcal{F'},\rhot,\rho)$ too.

    Indeed, the invariants of Definition~\ref{dfn:var-liftable} hold:
    \begin{itemize}
    \item Invariants~\ref{case:loc} and~\ref{case:term}: Immediate by the
    induction hypotheses and definition of tail and local positions.
    \item Invariant~\ref{case:exclu}: By the induction hypotheses,
    Invariant~\ref{case:loc}
    ($x$ is to appear in the new closure if and only if $f = h_i$).
    \item Invariant~\ref{case:share}: Lemma~\ref{lem:aliasing-letrec}
    (p.~\pageref{lem:aliasing-letrec}).
    \end{itemize}
    The other invariants hold trivially.

    By the induction hypotheses, we get
    \[\reductionF{\lift{b}}{s}{v}{s'}{\lift{\mathcal{F'}}}.\]
    If $f \neq h_i$,
    \[\lift{ \letrec{f(\range{x}{1}{n})}{a}{b} } = \letrec{f(\range{x}{1}{n})}{\lift{a}}{\lift{b}}\]
    hence, by definition of $\lift{\mathcal{F'}}$,
    \[\inferrule*[Left=(letrec)]{\
    \reductionF{\lift{b}}{s}{v}{s'}{\lift{\mathcal{F'}}} \\\\
    \rho' = \rhot\concat\rho|_{\dom(\rhot\concat\rho)\setminus\{\range{x}{1}{n}\}} \\
    \lift{\mathcal{F'}}=
    \ajout{\lift{\F}}{f}{\fun{\range{x}{1}{n}}{\lift{a}}{\rho',F}}    }{\
    \reduclift{\letrec{f(\range{x}{1}{n})}{a}{b}}{s}{v}{s'}}\]
    On the other hand, if $f = h_i$,
    \[\lift{ \letrec{f(\range{x}{1}{n})}{a}{b} } = \letrec{f(\range{x}{1}{n}x)}{\lift{a}}{\lift{b}}\]
    hence, by definition of $\lift{\mathcal{F'}}$,
    \[\inferrule*[Left=(letrec)]{\
    \reductionF{\lift{b}}{s}{v}{s'}{\lift{\mathcal{F'}}} \\\\
    \rho' = \rhot\concat\rho|_{\dom(\rhot\concat\rho)\setminus\{\range{x}{1}{n}x\}} \\
    \lift{\mathcal{F'}}=
    \ajout{\lift{\F}}{h_i}{\fun{\range{x}{1}{n}x}{\lift{a}}{\rho',F}}    }{\
    \reduclift{\letrec{h_i(\range{x}{1}{n})}{a}{b}}{s}{v}{s'}}\]

\paragraph{(val)}
    $\lift{v}=v$ so
    \[\inferrule*[Left=(val)]{ }{\reduclift{v}{s}{v}{\gc{\rhot}{s}}}\]

\paragraph{(var)}
    $\lift{y}=y$ so
    \[\inferrule*[Left=(var)]{\rhot\concat\rho\ y = l \in \dom\ s}{\reduclift{y}{s}{s\ l}{\gc{\rhot}{s}}}\]

\paragraph{(assign)}
    The parameter $x$ is liftable in $(y \coloneqq a,\F,\rhot,\rho)$ so in
    $(a,\F,\varepsilon,\rhot\concat\rho)$ too.

    Indeed, the invariants of Definition~\ref{dfn:var-liftable} hold:
    \begin{itemize}
    \item Invariant~\ref{case:share}: Lemma~\ref{lem:alias-concat},
    p.~\pageref{lem:alias-concat}.
    \end{itemize}
    The other invariants hold trivially.

    By the induction hypotheses, we get
    \[\reduclift[\env{}{\rhot\concat\rho}]{a}{s}{v}{s'}.\]
    Moreover
    \[\lift{y \coloneqq a}=y\coloneqq \lift{a},\]
    so :
    \[\inferrule*[Left=(assign)]{\reduclift[\env{}{\rhot\concat\rho}]{a}{s}{v}{s'}  \\
    \rhot\concat\rho\ y = l \in \dom\ s'}{\
    \reduclift{y \coloneqq a}{s}{\unit}{\gc{\rhot}{\subst{s'}{l}{v}}}}\]

\paragraph{(seq)}
    The parameter $x$ is liftable in $(a\ ;\ b,\F,\rhot,\rho)$.
    If $x$ is not defined in $a$ or $\F$, then $\lift{}$ is the identity
    function and can trivially be applied to the reduction of $a$.  Otherwise,
    $x$ is a liftable parameter in $(a,\F,\varepsilon,\rhot\concat\rho)$.

    Indeed, the invariants of Definition~\ref{dfn:var-liftable} hold:
    \begin{itemize}
    \item Invariant~\ref{case:share}: Lemma~\ref{lem:alias-concat},
    p.~\pageref{lem:alias-concat}.
    \end{itemize}
    The other invariants hold trivially.

    If $x$ is not defined in $b$ or $\F$, then $\lift{}$ is the identity
    function and can trivially be applied to the reduction of $b$.  Otherwise,
    $x$ is a liftable parameter in $(b,\F,\rhot,\rho)$.
    Indeed, the invariants of Definition~\ref{dfn:var-liftable} hold trivially.

    By the induction hypotheses, we get
    \reduclift[\env{}{\rhot\concat\rho}]{a}{s}{v}{s'} and
    \reduclift{b}{s'}{v'}{s''}.\\
    Moreover,
    \[\lift{a\ ;\ b} = \lift{a}\ ;\ \lift{b},\]
    hence:
    \[\inferrule*[Left=(seq)]{\reduclift[\env{}{\rhot\concat\rho}]{a}{s}{v}{s'} \\
    \reduclift{b}{s'}{v'}{s''}}{\
    \reduclift{a\ ;\ b}{s}{v'}{s''}}\]

\paragraph{(if-true) and (if-false)} are proved similarly to (seq).\qedhere

\newpage
\section{CPS conversion}\label{sec:cps-conversion}

In this section, we prove the correctness of the CPS-conversion performed by the
CPC translator.  This conversion is defined only on a subset of C programs that
we call \emph{CPS-convertible terms} (Section~\ref{sec:cps-convertible}).  We
first show that the \emph{early evaluation} of function parameters in CPS-convertible
terms is correct (Section~\ref{sec:early-eval}).  To simplify the proof of
correctness of CPS-conversion, we then introduce small-step reduction rules
featuring contexts and early evaluation (Section~\ref{sec:ss-reduction}).

In Section~\ref{sec:cps-terms}, we define \emph{CPS terms}, with the
\texttt{push} and \texttt{invoke} operators to build and execute continuations,
and the associated reduction rules.  Since the syntax of CPS-terms does not
ensure a correct reduction, we also define \emph{well-formed} CPS-terms, which
are the image of CPS-convertible terms by CPS-conversion.

The proof of correctness of CPS-conversion is finally carried out in
Section~\ref{sec:translation}. It consists merely in checking that the reduction
rules for CPS-convertible terms and well-formed CPS-terms execute in lock-step.

\subsection{CPS-convertible form}\label{sec:cps-convertible}

CPS conversion is not defined for every C function; instead, we restrict
ourselves to a subset of functions, which we call the {\em
    CPS-convertible\/} subset.
The CPS-convertible form restricts the calls to cps functions to
make it straightforward to capture their continuation.  In
CPS-convertible form, a call to a cps function \texttt{f} is either
in tail position, or followed by a tail call to another cps function
whose parameters are \emph{non-shared} variables that cannot be
modified by \texttt{f}.

In the C language, we define the CPS-convertible form as follows:
\begin{definition}[CPS-convertible form]\label{def:cps-c}
    A function {\tt h} is in \emph{CPS-convertible form} if every call to
    a cps function that it contains matches one of the following
    patterns, where both \texttt{f} and \texttt{g} are cps functions,
    \texttt{e$_\text{\tt 1}$, ..., e$_\text{\tt n}$} are any C expressions and
    \texttt{x, y$_\text{\tt 1}$, ..., y$_\text{\tt n}$} are distinct, non-shared
    variables:
    \begin{eqnarray}
        \text{\tt return f(e$_\text{\tt 1}$, ..., e$_\text{\tt n}$);}\\
        \text{\tt x = f(e$_\text{\tt 1}$, ..., e$_\text{\tt n}$); return g(x,
            y$_\text{\tt 1}$, ..., y$_\text{\tt n}$);}\\
        \text{\tt f(e$_\text{\tt 1}$, ..., e$_\text{\tt n}$); return g(x,
            y$_\text{\tt 1}$, ..., y$_\text{\tt n}$);}\label{useless1}\\
        \text{\tt f(e$_\text{\tt 1}$, ..., e$_\text{\tt n}$); return;}\\
        \text{\tt f(e$_\text{\tt 1}$, ..., e$_\text{\tt n}$); g(x,
            y$_\text{\tt 1}$, ..., y$_\text{\tt n}$); return;}\\
        \text{\tt x = f(e$_\text{\tt 1}$, ..., e$_\text{\tt n}$); g(x,
            y$_\text{\tt 1}$, ..., y$_\text{\tt n}$);
            return;}\label{useless2}
    \end{eqnarray}
    \qedhere
\end{definition}
Note the use of \texttt{return} to explicitly mark calls in tail position.  The
forms (\ref{useless1}) to (\ref{useless2}) are only necessary to
handle the cases where \texttt{f} and \texttt{g} return \texttt{void};
in the rest of the proof, we ignore these cases that are a syntactical
detail of the C language, and focus on the essential cases (1) and (2).

To prove the correctness of CPS-conversion, we need to express this definition
in our small imperative language.  This is done by defining CPS-convertible
terms, which are a subset of the terms introduced in
Definition~\ref{def:full-language} (Section~\ref{sec:definitions}).  A program
in CPS-convertible form consists of a set of mutually-recursive functions with
no free variables, the body of each of which is a CPS-convertible term.

A CPS\hyp{}convertible term has two parts: the head and the tail.  The head
is a (possibly empty) sequence of assignments, possibly embedded within
conditional statements.  The tail is a (possibly empty) sequence of
function calls in a highly restricted form: their parameters are
(side-effect free) expressions, except possibly for the last one, which
can be another function call of the same form.  Values and expressions
are left unchanged.

\begin{definition}[CPS-convertible terms]\label{def:cps-language}
 \begin{align*}
 v \Coloneqq & \quad\unit \;|\; \true \;|\; \false \;|\; n \in
 \mathbf{N}\tag{values}\\
 \expr \Coloneqq & \quad v \;|\; x\;|\; \ldots\tag{expressions}\\
    F \Coloneqq & f(\expr, \ldots, \expr) \;|\; f(\expr, \ldots, \expr, F)
    \tag{nested function calls}\\
    Q \Coloneqq & \epsilon \;|\; Q\ ;\ F \tag{tail}\\
    T \Coloneqq & \expr 
        \;|\; x \coloneqq \expr\ ;\ T
        \;|\; \ite{e}{T}{T}
        \;|\; Q\tag{head}
 \end{align*}
\end{definition}

The essential property of CPS\hyp{}convertible terms, which makes their CPS
conversion immediate to perform, is the guarantee that there is no cps
call outside of the tails.  It makes continuations easy to represent as
a series of function calls (tails) and separates them clearly from
imperative blocks (heads), which are not modified by the CPC translator.

The tails are a generalisation of Definition~\ref{def:cps-c}, which will be
useful for the proof of correctness of CPS-conversion.  Note that {\tt x =
    f(e$_\text{\tt 1}$, ..., e$_\text{\tt n}$); return g(x, y$_\text{\tt 1}$,
    ..., y$_\text{\tt n}$)} is represented by $g(f(\range{e}{1}{n}),
\range{y}{1}{n})$: this translation is correct because, contrary to C, our
language guarantees a left-to-right evaluation of function parameters.

Also noteworthy are the facts that:
\begin{itemize}
    \item there is no letrec construct anymore since every function is defined
        at top-level,
    \item assignments, conditions and function parameters of $f$ are restricted
        to expressions, to ensure that function calls only appear in tail
        position,
    \item there is no need to forbid shared variables in the parameters of $g$
        because they are ruled out of our language by design.
\end{itemize}

\subsection{Early evaluation}
\label{sec:early-eval}

In this section, we prove that correctness of \emph{early evaluation}, ie.\
evaluating the expressions $\expr$ before $F$ when reducing $f(\expr, \ldots,
\expr, F)$ in a tail.  This result is necessary to show the correctness of the
CPS-conversion, because function parameters are evaluated before any function
call when building continuations.

The reduction rules may be simplified somewhat for CPS-convertible terms.  We do
not need to keep an explicit environment of functions since there are no inner
functions any more; for the same reason, the (letrec) rule disappears.  Instead,
we use a constant environment $\F$ holding every function used in the reduced
term $M$.  To account for the absence of free variables, the closures in $\F$
need not carry an environment.  As a result, in the (call) rule,
$\rho'=\varepsilon$ and $\mathcal{F'}=\F$.

Early evaluation is correct for lifted terms because a lifted term can never
modify the variables that are not in its environment, since it cannot access
them through closures.
\begin{lemma}\label{lem:lift-store-invariant}
  Let $M$ be a lambda-lifted term. Then,
  \[\reductionN{M}{s}{v}{s'}\]
  implies
  \[s|_{\dom(s)\setminus\Image(\rho)}=s'|_{\dom(s)\setminus\Image(\rho)}.\]
\end{lemma}
\begin{proof}
  By induction on the structure of the reduction.  The key points are
  the use of $\rho'=\varepsilon$ in the (call) case, and the absence of
  (letrec) rules.

\paragraph{(val) and (var)} Trivial ($s = s'$).

\paragraph{(assign)} By the induction hypotheses, 
  \[s|_{\dom(s)\setminus\Image(\rho)}=s'|_{\dom(s)\setminus\Image(\rho)}
  \text{ and } l \in \Image(\rho),\] hence
  \[s|_{\dom(s)\setminus\Image(\rho)}=(\subst{s'}{l}{v})|_{\dom(s)\setminus\Image(\rho)}.\]

\paragraph{(seq)} By the induction hypotheses,
  \[s|_{\dom(s)\setminus\Image(\rho)}=s'|_{\dom(s)\setminus\Image(\rho)}
  \text{ and }
  s'|_{\dom(s')\setminus\Image(\rho)}=s''|_{\dom(s')\setminus\Image(\rho)}.\]
  Since,
  $\dom(s)\subset\dom(s')$, the second equality can be restricted to
  \[s'|_{\dom(s)\setminus\Image(\rho)}=s''|_{\dom(s)\setminus\Image(\rho)}.\]
  Hence,
  \[s|_{\dom(s)\setminus\Image(\rho)}=s''|_{\dom(s)\setminus\Image(\rho)}.\]

\paragraph{(if-true) and (if-false)} are proved similarly to (seq).

\paragraph{(letrec)} doesn't occur since $M$ is lambda-lifted.

\paragraph{(call)} By the induction hypotheses,
  \[
  (\subst{s_{n+1}}{l_i}{v_i})|_{\dom(\subst{s_{n+1}}{l_i}{v_i})\setminus\Image(\rho''\concat\rho')}=
    s'|_{\dom(\subst{s_{n+1}}{l_i}{v_i})\setminus\Image(\rho''\concat\rho')}\]
  Since $\rho' = \varepsilon$, $\Image(\rho'') = \{l_i\}$ and
  $\dom(s_{n+1})\cap\{l_i\} = \emptyset$ (by freshness),
  \[
  (\subst{s_{n+1}}{l_i}{v_i})|_{\dom(s_{n+1})}=
    s'|_{\dom(s_{n+1})}\]
  so $s_{n+1} = s'|_{\dom(s_{n+1})}$.\\
  Since $\dom(s)\setminus\Image(\rho)\subset\dom(s)\subset\dom(s_{n+1})$,
  \[s_{n+1}|_{\dom(s)\setminus\Image(\rho)} = s'|_{\dom(s)\setminus\Image(\rho)}.\]
  Finally, we can prove similarly to the (seq) case that
  \[s|_{\dom(s)\setminus\Image(\rho)}=s_{n+1}|_{\dom(s)\setminus\Image(\rho)}.\]
  Hence,
  \[s|_{\dom(s)\setminus\Image(\rho)}=s'|_{\dom(s)\setminus\Image(\rho)}. 
      \qedhere\]
\end{proof}

As a consequence, a tail of function calls cannot modify the current
store, only extend it with the parameters of the called functions.
\begin{corollary}\label{cor:lift-store-invariant}
For every tail $Q$,
\[\reductionN{Q}{s}{v}{s'}
\text{ implies }
s \revextends s'.\]
\end{corollary}
\begin{proof}
We prove the corollary by induction on the structure of a tail.
First remember that \emph{store extension} (written $\revextends$) is a
partial order over stores (Property~\ref{prop:ext-order}), defined in
Section~\ref{subsec:second-step} as follows: $s \revextends s' \ssi
s'|_{\dom(s)} = s$.

The case $\epsilon$ is  trivial.  The case $Q\ ;\ F$ is immediate by
induction ((seq) rule), since $\revextends$ is transitive.  Similarly,
it is pretty clear that $f(\expr, \ldots, \expr, F)$ follows by
induction and transitivity from $f(\expr, \ldots, \expr)$ ((call) rule).
We focus on this last case.

Lemma~\ref{lem:lift-store-invariant} implies:
  \[
  (\subst{s_{n+1}}{l_i}{v_i})|_{\dom(\subst{s_{n+1}}{l_i}{v_i})\setminus\Image(\rho''\concat\rho')}=
    s'|_{\dom(\subst{s_{n+1}}{l_i}{v_i})\setminus\Image(\rho''\concat\rho')}.\]
  Since $\rho' = \varepsilon$, $\Image(\rho'') = \{l_i\}$ and
  $\dom(s_{n+1})\cap\{l_i\} = \emptyset$ (by freshness),
  \[
  (\subst{s_{n+1}}{l_i}{v_i})|_{\dom(s_{n+1})}=
    s'|_{\dom(s_{n+1})}\]
  so $s_{n+1} = s'|_{\dom(s_{n+1})}$.

The evaluation of $\expr$ parameters do not change the store:
$s_{n+1}=s$.  The expected result follows: $s = s'|_{\dom(s)}$, hence
$s \revextends s'$.
\end{proof}

This leads to the correctness of early evaluation.
\begin{theorem}[Early evaluation]
\label{thm:preevaluation}
For every tail $Q$,
$\reductionN{Q}{s}{v}{s'}$ implies
$\reductionN{Q[x\setminus s(\rho\ x)]}{s}{v}{s'}$ (provided $x\in\dom(\rho)$ and
$\rho\ x\in\dom(s)$).
\end{theorem}
\begin{proof}
Immediate induction on the structure of tails and expressions:
Corollary~\ref{cor:lift-store-invariant} implies that
$s \revextends s''$ and $\rho\ x\in\dom(s)$ ensures
that $s(\rho\ x) = s''(\rho\ x)$ in the relevant cases (namely the (seq) rule
for $Q\ ;\ F$ and the (call) rule for $f(\expr, \ldots, \expr, F)$).
\end{proof}

\subsection{Small-step reduction}
\label{sec:ss-reduction}

We define the semantics of CPS\hyp{}convertible terms through a set of
small-step reduction rules.  We distinguish three kinds of reductions:
$\rightarrow_T$ to reduce the head of terms, $\rightarrow_Q$ to reduce
the tail, and $\rightarrow_e$ to evaluate expressions.

These rules describe a stack machine with a store $\sigma$ to keep the
value of variables.  Since free and shared variables have been
eliminated in earlier passes, there is a direct correspondence at any
point in the program between variable names and locations, with no need
to dynamically maintain an extra environment.

We use contexts as a compact representation for stacks.  The head rules
$\rightarrow_T$ reduce triples made of a term, a context and a store:
$\langle T, C[\ ], \sigma \rangle$.  The tail rules $\rightarrow_Q$,
which merely unfold tails with no need of a store, reduce couples of a
tail and a context:  $\langle Q, C[\ ], \rangle$.  The expression rules
do not need context to reduce, thus operating on couples made of an
expression and a store: $\langle e, \sigma \rangle$.

\paragraph{Contexts}
Contexts are sequences of function calls.  In those sequences, function
parameters shall be already evaluated: constant expressions are allowed,
but not variables.  As a special case, the last parameter might be a
``hole'' instead, written $\circleddash$, to be filled with the return
value of the next, nested function.

\begin{definition}[Contexts] Contexts are defined inductively:
    \[C \Coloneqq [\ ]
    \;|\; C[[\ ]\ ;\ f(v, \ldots, v)]
    \;|\; C[[\ ]\ ;\ f(v, \ldots, v, \circleddash)]
    \]
\end{definition}

\begin{definition}[CPS\hyp{}convertible reduction rules]
\begin{align}
\langle x \coloneqq \expr\ ;\ T, C[\ ], \sigma \rangle &\rightarrow_T
\langle T, C[\ ], \sigma[x\mapsto v]  \rangle
\\&\quad\text{when $\langle \expr,\sigma \rangle \rightarrow_e^\star
    v$}\notag\\
\langle \ite{\expr}{T_1}{T_2}, C[\ ], \sigma \rangle &\rightarrow_T
\langle T_1, C[\ ], \sigma \rangle
\\&\quad\text{when }\langle \expr,\sigma \rangle \rightarrow_e^\star
    \true\notag\\
\langle \ite{\expr}{T_1}{T_2}, C[\ ], \sigma \rangle &\rightarrow_T
\langle T_2, C[\ ], \sigma \rangle
\\&\quad\text{when }\langle \expr,\sigma \rangle \rightarrow_e^\star
    \false\notag\\
\langle \expr, C[[\ ]\ ;\ f(v_1, \ldots, v_n)], \sigma \rangle &\rightarrow_T
\langle \epsilon, C[[\ ]\ ;\ f(v_1, \ldots, v_n)] \rangle\\
\langle \expr, C[[\ ]\ ;\ f(v_1, \ldots, v_n, \circleddash)], \sigma
\rangle &\rightarrow_T
\langle \epsilon, C[[\ ]\ ;\ f(v_1, \ldots, v_n, v)] \rangle
\\&\quad\text{when $\langle \expr,\sigma \rangle \rightarrow_e^\star
    v$}\notag\\
\langle \expr, [\ ], \sigma \rangle &\rightarrow_T v \quad\text{when $\langle \expr,\sigma \rangle \rightarrow_e^\star
    v$}\notag\\
\langle Q, C[\ ], \sigma \rangle &\rightarrow_T
\langle Q[x_i\setminus \sigma\ x_i], C[\ ] \rangle
\label{rule:preevaluation}
\\&\quad\text{for every $x_i$ in $\dom(\sigma)$}\notag\\
\notag\\
\langle Q\ ;\ f(v_1, \ldots, v_n), C[\ ] \rangle &\rightarrow_Q
\langle Q, C[[\ ]\ ;\ f(v_1, \ldots, v_n)] \rangle
\label{rule:context}
\\
\langle Q\ ;\ f(v_1, \ldots, v_n, F), C[\ ] \rangle &\rightarrow_Q
\langle Q\ ;\ F, C[[\ ]\ ;\ f(v_1, \ldots, v_n, \circleddash)] \rangle
\label{rule:context-hole}
\\
\langle \epsilon, C[[\ ]\ ;\ f(v_1, \ldots, v_n)] \rangle &\rightarrow_Q
\langle T, C[\ ], \sigma \rangle
\label{rule:exec}
\\&\quad\text{when } f(x_1,\ldots,x_n) = T
    \text{ and }\sigma = \{x_i\mapsto v_i\}\notag
\end{align}
We do not detail the rules for $\rightarrow_e$, which simply looks for
variables in $\sigma$ and evaluates arithmetical and boolean operators.
\end{definition}

\paragraph{Early evaluation}
Note that Rule~\ref{rule:preevaluation} evaluates every function parameter in a
tail before the evaluation of the tail itself.  This is precisely the early
evaluation process described above, which is correct by
Theorem~\ref{thm:preevaluation}.  We introduce early evaluation directly in the
reduction rules rather than using it as a lemma to simplify the proof of
correctess of the CPS-conversion.

\subsection{CPS terms}
\label{sec:cps-terms}

Unlike classical CPS conversion techniques \cite{plotkin}, our CPS terms
are not continuations, but a procedure which builds and executes the
continuation of a term.  Construction is performed by $\push$, which
adds a function to the current continuation, and execution by $\invoke$,
which calls the first function of the continuation, optionally passing
it the return value of the current function.

\begin{definition}[CPS terms]
 \begin{align*}
 v \Coloneqq & \quad\unit \;|\; \true \;|\; \false \;|\; n \in
 \mathbf{N}\tag{values}\\
 \expr \Coloneqq & \quad v \;|\; x\;|\; \ldots\tag{expressions}\\
    Q \Coloneqq & \invoke
        \;|\; \push\ f(\expr, \ldots, \expr)\ ;\ Q
        \;|\; \push\ f(\expr, \ldots, \expr, \boxdot)\ ;\ Q\tag{tail}\\
    T \Coloneqq & \patch\ \expr 
        \;|\; x \coloneqq \expr\ ;\ T
        \;|\; \ite{e}{T}{T}
        \;|\; Q\tag{head}
 \end{align*}
\end{definition}

\paragraph{Continuations and reduction rules}
A continuation is a sequence of function calls to be performed, with
already evaluated parameters.  We write $\cconcat$ for appending a
function to a continuation, and $\boxdot$ for a ``hole'', i.e.\ an
unknown parameter.

\begin{definition}[Continuations]
    \[\mathcal{C} \Coloneqq \varepsilon
    \;|\; \ f(v, \ldots, v) \cconcat \mathcal{C}
    \;|\; \ f(v, \ldots, v, \boxdot) \cconcat \mathcal{C}
    \]
\end{definition}

The reduction rules for CPS terms are isomorphic to the rules for
CPS\hyp{}convertible terms, except that they use continuations instead of
contexts.

\begin{definition}[CPS reduction rules]
{\allowdisplaybreaks
\begin{align}
\langle x \coloneqq \expr\ ;\ T, \mathcal{C}, \sigma \rangle
&\rightarrow_T
\langle T, \mathcal{C}, \sigma[x\mapsto v] \rangle
\\&\quad\text{when $\langle \expr,\sigma \rangle \rightarrow_e^\star
    v$}\notag\\
\langle \ite{\expr}{T_1}{T_2}, \mathcal{C}, \sigma \rangle
&\rightarrow_T
\langle T_1, \mathcal{C}, \sigma \rangle
\\&\quad\text{if }\langle \expr,\sigma \rangle \rightarrow_e^\star
    \true\notag\\
\langle \ite{\expr}{T_1}{T_2}, \mathcal{C}, \sigma \rangle
&\rightarrow_T
\langle T_2, \mathcal{C}, \sigma \rangle
\\&\quad\text{if }\langle \expr,\sigma \rangle \rightarrow_e^\star
    \false\notag\\
\langle \patch\ \expr, f(v_1, \ldots, v_n) \cconcat \mathcal{C},
\sigma \rangle &\rightarrow_T
\langle \invoke, f(v_1, \ldots, v_n) \cconcat \mathcal{C} \rangle\\
\langle \patch\ \expr, f(v_1, \ldots, v_n, \boxdot) \cconcat
\mathcal{C}, \sigma \rangle &\rightarrow_T
\langle \invoke, f(v_1, \ldots, v_n, v) \cconcat \mathcal{C} \rangle
\\&\quad\text{when $\langle \expr,\sigma \rangle \rightarrow_e^\star
    v$}\notag\\
\langle \patch\ \expr, \varepsilon, \sigma \rangle &\rightarrow_T v \quad\text{when $\langle \expr,\sigma \rangle \rightarrow_e^\star
    v$}\notag\\
\langle Q, \mathcal{C}, \sigma \rangle &\rightarrow_T
\langle Q[x_i\setminus \sigma\ x_i], \mathcal{C} \rangle
\\&\quad\text{for every $x_i$ in $\dom(\sigma)$}\notag\\
\notag\\
\langle \push\ f(v_1, \ldots, v_n)\ ;\ Q, \mathcal{C} \rangle
&\rightarrow_Q
\langle Q, f(v_1, \ldots, v_n) \cconcat \mathcal{C} \rangle\\
\langle \push\ f(v_1, \ldots, v_n, \boxdot)\ ;\ Q, \mathcal{C} \rangle
&\rightarrow_Q
\langle Q, f(v_1, \ldots, v_n, \boxdot) \cconcat \mathcal{C} \rangle\\
\langle \invoke, f(v_1, \ldots, v_n) \cconcat \mathcal{C} \rangle
&\rightarrow_Q
\langle T, \mathcal{C}, \sigma \rangle
\\&\quad\text{when } f(x_1,\ldots,x_n) = T
    \text{ and }\sigma = \{x_i\mapsto v_i\}\notag
\end{align}
}
\end{definition}

\paragraph{Well-formed terms}
Not all CPS term will lead to a correct reduction.  If we $\push$ a
function expecting the result of another function and $\invoke$ it
immediately, the reduction blocks:
\[\langle \push\ f(v_1, \ldots, v_n, \boxdot)\ ;\ \invoke, \mathcal{C}, \sigma \rangle \rightarrow
\langle \invoke, f(v_1, \ldots, v_n, \boxdot) \cconcat \mathcal{C}, \sigma
\rangle \not\rightarrow\]

\emph{Well-formed terms} avoid this behaviour.
\begin{definition}[Well-formed term]
A continuation queue is \emph{well-formed} if it does not end with:
\[\push\ f(\expr,\ldots, \expr, \boxdot)\ ;\ \invoke.\]

A term is \emph{well-formed} if every continuation queue in this term is
well-formed.
\end{definition}

\subsection{Correctess of the CPS-conversion}\label{sec:translation}

We define the CPS conversion as a mapping from CPS\hyp{}convertible terms to
CPS terms.
\begin{definition}[CPS conversion]
\begin{align*}
 (Q\ ;\ f(\expr, \ldots, \expr))^\blacktriangle &= \push\ f(\expr, \ldots, \expr)\ ;\
 Q^\blacktriangle\\
 (Q\ ;\ f(\expr, \ldots, \expr, F))^\blacktriangle &=  \push\ f(\expr, \ldots, \expr,
 \boxdot)\ ;\ (Q\ ;\ F)^\blacktriangle\\
 \epsilon^\blacktriangle &= \invoke\\
 (x \coloneqq \expr\ ;\ T)^\blacktriangle &= x \coloneqq \expr\ ;\ T^\blacktriangle\\
 (\ite{\expr}{T_1}{T_2})^\blacktriangle &= \ite{\expr}{T_1^\blacktriangle}{T_2^\blacktriangle}\\
 \expr^\blacktriangle &= \patch\ \expr
\end{align*}
\end{definition}

In the rest of this section, we prove that this mapping yields an
isomorphism between the reduction rules of CPS\hyp{}convertible terms and
well-formed CPS terms, whence the correctness of our CPS conversion
(Theorem~\ref{thm:cps-correct}).

We first prove two lemmas to show that $^\blacktriangle$ yields only
well-formed CPS terms.  This leads to a third lemma to show that
$^\blacktriangle$ is a bijection between CPS\hyp{}convertible terms and
well-formed CPS terms.

CPS\hyp{}convertible terms have been carefully designed to make CPS conversion
as simple as possible.  Accordingly, the following three proofs, while long
and tedious, are fairly trivial.

\begin{lemma}
Let $Q$ be a continuation queue. Then $Q^\blacktriangle$ is well-formed.
\end{lemma}
\begin{proof}
By induction on the structure of a tail.
\[
 \epsilon^\blacktriangle = \invoke
 \]
 and
\[
(\epsilon\ ;\ f(\expr, \ldots, \expr))^\blacktriangle = \push\ f(\expr, \ldots, \expr)\ ;\
 \invoke
 \]
 are well-formed by definition.
\[
((Q\ ;\ F)\ ;\ f(\expr, \ldots, \expr))^\blacktriangle = \push\ f(\expr, \ldots, \expr)\ ;\
 (Q\ ;\ F)^\blacktriangle
 \]
and
\[
 (Q\ ;\ f(\expr, \ldots, \expr, F))^\blacktriangle =  \push\ f(\expr, \ldots, \expr,
 \boxdot)\ ;\ (Q\ ;\ F)^\blacktriangle
\]
are well-formed by induction.
\end{proof}

\begin{lemma}
Let $T$ be a CPS\hyp{}convertible term. Then $T^\blacktriangle$ is well-formed.
\end{lemma}
\begin{proof}
Induction on the structure of $T$, using the above lemma.
\end{proof}

\begin{lemma}\label{lemma:cps-iso}
The $^\blacktriangle$ relation is a bijection between CPS\hyp{}convertible terms
and well-formed CPS terms.
\end{lemma}
\begin{proof}
Consider the following mapping from well-formed CPS terms to CPS\hyp{}convertible
terms:
\begin{align*}
 (\push\ f(\expr, \ldots, \expr)\ ;\ Q)^\blacktriangledown &=
 Q^\blacktriangledown\ ;\ f(\expr, \ldots, \expr)\\
 (\push\ f(\expr, \ldots, \expr, \boxdot)\ ;\ Q)^\blacktriangledown &=
 Q'\ ;\ f(\expr, \ldots, \expr, F)\\
 &\quad\text{with $Q^\blacktriangledown = Q'\ ;\ F$}\tag{*}\\
 \invoke^\blacktriangledown &= \epsilon\\
 (x \coloneqq \expr\ ;\ T)^\blacktriangledown &= x \coloneqq \expr\ ;\ T^\blacktriangledown\\
 \ite{\expr}{T_1}{T_2}^\blacktriangledown &= \ite{\expr}{T_1^\blacktriangledown}{T_2^\blacktriangledown}\\
 (\patch\ \expr)^\blacktriangledown &= \expr
\end{align*}
(*) The existence of $Q'$ is guaranteed by well-formedness:
\begin{itemize}
\item $\forall T,\ T^\blacktriangledown=\epsilon\ \Rightarrow\ T=\invoke$ (by
disjunction on the definition of~$^\blacktriangledown$),
\item here, $Q\neq\invoke$ because $(\push\ f(\expr, \ldots, \expr, \boxdot)\ ;\
Q)$ is well-formed,
\item hence $Q^\blacktriangledown\neq\epsilon$.
\end{itemize}
One checks easily that $(T^\blacktriangledown)^\blacktriangle=T$ and
$(T^\blacktriangle)^\blacktriangledown=T$.
\end{proof}

To conclude the proof of isomorphism, we also need an (obviously
bijective) mapping from contexts to continuations:
\begin{definition}[Conversion of contexts]
\begin{align*}
([\ ])^\vartriangle &= \varepsilon\\
(C[[\ ]\ ;\ f(v_1, \ldots, v_n)])^\vartriangle &=
    f(v_1, \ldots, v_n) \cconcat \mathcal{C}
\\&\quad\text{with $(C[\ ])^\vartriangle = \mathcal{C}$}\\
(C[[\ ]\ ;\ f(v_1, \ldots, v_n, \circleddash)])^\vartriangle &=
    f(v_1,\ldots, v_n, \boxdot) \cconcat \mathcal{C}
\\&\quad\text{with $(C[\ ])^\vartriangle = \mathcal{C}$}
\end{align*}
\end{definition}

The correctness theorem follows:
\begin{theorem}[Correctness of CPS conversion]\label{thm:cps-correct}
  The $^\blacktriangle$ and $^\vartriangle$ mappings are two bijections,
  the inverses of which are written $^\blacktriangledown$ and
  $^\triangledown$.  They yield an isomorphism between reduction rules
  of CPS\hyp{}convertible terms and CPS terms.
\end{theorem}
\begin{proof}
Lemma~\ref{lemma:cps-iso} ensures that $^\blacktriangle$ is a bijection
between CPS\hyp{}convertible terms and well-formed CPS terms.  Moreover,
$^\vartriangle$ is an obvious bijection between contexts and
continuations.

To complete the proof, we only need to apply $^\blacktriangle$,
$^\vartriangle$, $^\blacktriangledown$ and $^\triangledown$ to
CPS\hyp{}convertible terms, contexts, well-formed CPS terms and continuations
(respectively) in every reduction rule and check that we get a valid
rule in the dual reduction system.  The result is summarized in
Figure~\ref{fig:iso}.
\end{proof}

\begin{landscape}
\begin{figure}
\begin{align*}
\langle x \coloneqq \expr\ ;\ T, C[\ ], \sigma \rangle &\rightarrow_T
\langle T, C[\ ], \sigma[x\mapsto v] \rangle
&\Leftrightarrow&&
\langle x \coloneqq \expr\ ;\ T, \mathcal{C}, \sigma \rangle &\rightarrow_T
\langle T, \mathcal{C}, \sigma[x\mapsto v] \rangle
\\&&&&&\quad\text{when $\langle \expr,\sigma \rangle \rightarrow_e^\star
    v$}\notag\\
%
\langle \ite{\expr}{T_1}{T_2}, C[\ ], \sigma \rangle &\rightarrow_T
\langle T_1, C[\ ], \sigma \rangle
&\Leftrightarrow&&
\langle \ite{\expr}{T_1}{T_2}, \mathcal{C}, \sigma \rangle &\rightarrow_T
\langle T_1, \mathcal{C}, \sigma \rangle
\\&&&&&\quad\text{if }\langle \expr,\sigma \rangle \rightarrow_e^\star
    \true\notag\\
\langle \ite{\expr}{T_1}{T_2}, C[\ ], \sigma \rangle &\rightarrow_T
\langle T_2, C[\ ], \sigma \rangle
&\Leftrightarrow&&
\langle \ite{\expr}{T_1}{T_2}, \mathcal{C}, \sigma \rangle &\rightarrow_T
\langle T_2, \mathcal{C}, \sigma \rangle
\\&&&&&\quad\text{if }\langle \expr,\sigma \rangle \rightarrow_e^\star
    \false\notag\\
%
\langle \expr, C[[\ ]\ ;\ f(v_1, \ldots, v_n)], \sigma \rangle &\rightarrow_T
\langle \epsilon, C[[\ ]\ ;\ f(v_1, \ldots, v_n)] \rangle
&\Leftrightarrow&&
\langle \patch\ \expr, f(v_1, \ldots, v_n) \cconcat \mathcal{C}, \sigma \rangle &\rightarrow_T
\langle \invoke, f(v_1, \ldots, v_n) \cconcat \mathcal{C} \rangle\\
%
\langle \expr, C[[\ ]\ ;\ f(v_1, \ldots, v_n, \circleddash)], \sigma \rangle &\rightarrow_T
\langle \epsilon, C[[\ ]\ ;\ f(v_1, \ldots, v_n, v)] \rangle
&\Leftrightarrow&&
\langle \patch\ \expr, f(v_1, \ldots, v_n, \boxdot) \cconcat \mathcal{C}, \sigma \rangle &\rightarrow_T
\langle \invoke, f(v_1, \ldots, v_n, v) \cconcat \mathcal{C} \rangle
\\&&&&&\quad\text{when $\langle \expr,\sigma \rangle \rightarrow_e^\star
    v$}\notag\\
%
\langle \expr, [\ ], \sigma \rangle &\rightarrow_T v
&\Leftrightarrow&&
\langle \patch\ \expr, \varepsilon, \sigma \rangle &\rightarrow_T v
\\&&&&&\quad\text{when $\langle \expr,\sigma \rangle \rightarrow_e^\star
    v$}\notag\\
%
\langle Q, C[\ ], \sigma \rangle &\rightarrow_T
\langle Q[x_i\setminus \sigma\ x_i], C[\ ] \rangle
&\Leftrightarrow&&
\langle Q, \mathcal{C}, \sigma \rangle &\rightarrow_T
\langle Q[x_i\setminus \sigma\ x_i], \mathcal{C} \rangle
\\&&&&&\quad\text{for every $x_i$ in $\dom(\sigma)$}\notag\\
\notag\\
\langle Q\ ;\ f(v_1, \ldots, v_n), C[\ ] \rangle &\rightarrow_Q
\langle Q, C[[\ ]\ ;\ f(v_1, \ldots, v_n)] \rangle
&\Leftrightarrow&&
\langle \push\ f(v_1, \ldots, v_n)\ ;\ Q, \mathcal{C} \rangle &\rightarrow_Q
\langle Q, f(v_1, \ldots, v_n) \cconcat \mathcal{C} \rangle\\
%
\langle Q\ ;\ f(v_1, \ldots, v_n, F), C[\ ] \rangle &\rightarrow_Q
\langle Q\ ;\ F, C[[\ ]\ ;\ f(v_1, \ldots, v_n, \circleddash)] \rangle
&\Leftrightarrow&&
\langle \push\ f(v_1, \ldots, v_n, \boxdot)\ ;\ Q', \mathcal{C} \rangle &\rightarrow_Q
\langle Q', f(v_1, \ldots, v_n, \boxdot) \cconcat \mathcal{C} \rangle
\\&&&&&\quad\text{when $Q' = (Q\ ;\ F)^\blacktriangle$}\\
%
\langle \epsilon, C[[\ ]\ ;\ f(v_1, \ldots, v_n)] \rangle &\rightarrow_Q
\langle T, C[\ ], \sigma \rangle
&\Leftrightarrow&&
\langle \invoke, f(v_1, \ldots, v_n) \cconcat \mathcal{C} \rangle &\rightarrow_Q
\langle T, \mathcal{C}, \sigma \rangle
\\&&&&&\quad\text{when } f(x_1,\ldots,x_n) = T
    \text{ and }\sigma = \{x_i\mapsto v_i\}\notag
\end{align*}
\caption{Isomorphism between reduction rules}\label{fig:iso}
\end{figure}
\end{landscape}

\end{document}